\documentclass[aoas]{imsart}

\usepackage{amssymb}
\usepackage{amsmath}
\usepackage{stmaryrd}
\usepackage{ulem}
\usepackage{bm}

\RequirePackage{amsthm,amsmath,amsfonts,amssymb}
\RequirePackage[authoryear]{natbib}
\RequirePackage[colorlinks,citecolor=blue,urlcolor=blue]{hyperref}
\RequirePackage{graphicx}

\startlocaldefs

\DeclareMathOperator{\dd}{d}

\DeclareMathOperator*{\argmax}{arg max}
\DeclareMathOperator*{\argmin}{arg min}

\theoremstyle{plain}

\newtheorem{theorem}{Theorem}[section]
\newtheorem{lemma}[theorem]{Lemma}
\newtheorem*{lemma*}{Lemma}
\newtheorem*{theorem*}{Theorem}
\newtheorem{remark}[theorem]{Remark}
\theoremstyle{remark}
\newtheorem{definition}[theorem]{Definition}

\endlocaldefs

\begin{document}

\begin{frontmatter}
\title{A continuous multiple hypothesis testing framework for optimal exoplanet detection}
\runtitle{Continuous multiple hypothesis testing for optimal exoplanet detection}

\begin{aug}
\author[A]{\fnms{Nathan C.} \snm{Hara}\ead[label=e1]{nathan.hara@unige.ch}\ead[label=e2]{nathan.hara@lam.fr}},
\author[B]{\fnms{Thibault} \snm{de Poyferr\'e}},
\author[A]{\fnms{Jean-Baptiste} \snm{Delisle}} and
\author[C]{\fnms{Marc} \snm{Hoffmann}}


\address[A]{Observatoire Astronomique de l’Université de Genève, 51 Chemin de Pegasi b, 1290 Versoix, Switzerland, \printead{e1}, \printead{e2}}
\address[B]{Mathematical Science Research Institute, Berkeley, USA}
\address[C]{University Paris-Dauphine, CEREMADE, Place du Mar\'echal De Lattre de Tassigny, 75016 Paris, France}
\end{aug}

\begin{abstract}

When searching for exoplanets, one wants to count how many planets orbit a given star, and to determine what their characteristics are. If the estimated planet characteristics are too far from those of a planet truly present, this should be considered as a false detection. This setting is a particular instance of a general one: aiming to retrieve parametric components in a dataset corrupted by nuisance signals, with a certain accuracy on their parameters. We exhibit a detection criterion minimizing false and missed detections, either as a function of their relative cost or when the expected number of false detections is bounded. If the components can be separated in a technical sense discussed in detail, the optimal detection criterion is a posterior probability obtained as a by-product of Bayesian evidence calculations. Optimality is guaranteed within a model, and we introduce model criticism methods to ensure that the criterion is robust to model errors. We show on two simulations emulating exoplanet searches that the optimal criterion can significantly outperform other criteria. Finally, we show that our framework offers solutions for the identification of components of mixture models and Bayesian false discovery rate control when hypotheses are not discrete.

\end{abstract}

\begin{keyword}
\kwd{Exoplanets}
\kwd{Bayesian}
\kwd{False discovery rate}
\kwd{Decision theory}
\kwd{Maximum utility}
\end{keyword}

\end{frontmatter}

\section{Introduction}
\label{sec:introduction}

Planets outside the Solar System, or exoplanets, can be detected by several observational techniques, leading to different types of data~\citep{perryman2018}. In all cases, based on these data, the goal is to determine how many planets can be confidently detected and what their characteristics are: orbital period, eccentricity, mass, radius, etc. Even if the number of planets detected matches the true number of planets, if the orbital elements of a planet whose detection is claimed are too far from those of a planet truly in the system, it might lead to incorrect scientific conclusions. We thus need to tie the detection criterion to the desired accuracy on the orbital elements.

 In the present work, we adopt a general parametric model encompassing exoplanets. The data is described by $n = 0$ to $n_{\mathrm{max}}$ parametrized signals, which we call components, as well as parametrized nuisance signals. In the case of exoplanets, the parameters of component $i$ are the orbital elements of planets $i$, and the nuisance signals parameters include instrument offsets (intercepts), linear trends due to stellar companions, parameters of the noise probability distribution, etc. Our aim is to determine how many components are present, with a desired accuracy on their parameters.
 
The way we formalize this problem can be seen as an extension to the continuous case of the discrete Bayesian false discovery rate problem of~\cite{muller2004, muller2006}. 
They consider a dataset that can potentially support $n$ out of $m$ discrete hypotheses $H_i$, $i=1..m$. Their goal is to find $n$, and a subset $S$ of $n$ indices such that the claim ``the $n$ hypotheses $H_i, i\in S$ are accepted as true'' is optimal in a certain sense. In our work, hypotheses are of the form: ``there are $n$ components, one with parameters in $\Theta_1$, one with parameters in $\Theta_2$..., one with parameters in $\Theta_n$'', where $\Theta_i$s are regions of the parameter space. If, for instance, the $\Theta_i$s are balls of fixed radius, their centre can be represented by a vector of continuous parameters, thus, informally speaking, as continuous indices.

 If among the chosen $(\Theta_i)_{i=1..n}$,  $p$ of them do not contain the parameters of a component truly present, we count $p$ false detections. Conversely, if there are $q$ components truly present whose parameters do not belong to one of the $\Theta_i$s, we count $q$ missed detections.  
 To optimize the choice of $(\Theta_i)_{i=1..n}$, as in~\cite{muller2004}, we adopt two Bayesian approaches. First, we compute the maximum utility action, where the objective function performs a trade-off between expected false and missed detections. Second, we minimise the expected number of missed detections subject to a constraint on the expected number of false ones. 
We find that provided the $\Theta_i$s cannot be too close to each other in some sense, the optimal detection procedure simply consists in taking the $\Theta_i, i=1..n$ such that for each $i$, the posterior probabilities of having a component with parameters in $\Theta_i$ is greater than a certain threshold.

 

Our analysis can be applied to any exoplanet observation technique. However, to illustrate our method, we focus on the radial velocity (RV) technique, an exoplanet detection method poised to play a crucial role in the coming decades, in particular in the search for life outside the Solar System~\citep{nasaeprv2021}. For the sake of simplicity, we consider here the last step of RV analysis: finding parametric periodic signals in an unevenly sampled time series, and we refer the reader to~\cite{haraford2023} for a more in-depth presentation of RV data analysis. For this particular step, several techniques exist: the Bayes factor \citep[e.g.][]{ford2007, gregory2007, tuomi2009, brewerdonovan2015, diaz2016, faria2016, nelson2020}, periodograms~\citep{jaynes2003, cumming2004, baluev2008, zechmeister2009, baluev2009, baluev2013, baluev2015, delisle2019a}, or sparse recovery techniques associated to a false alarm probability~\citep{hara2017}. As argued in~\cite{hara2021}, these statistical significance metrics do not explicitly encode the desired accuracy on the orbital elements. 
The optimal detection criterion does encode it, and we show on a simulation that it can significantly outperform existing criteria.

As a posterior probability, our significance metric is particularly meaningful if it is calibrated. That is, if among the detections with posterior probability $\alpha$, on average a fraction $\alpha$ of them is correct. Using the vocabulary of~\cite{box1980}, finding an optimal decision pertains to the ``estimation'' problem, that is selecting a decision among a set of alternatives. However, calibrating the significance metric pertains to ``model criticism'': determining whether this set of alternatives is reasonable. We define a general test for model criticism based on a Bayesian cross-validation. We apply it in the case of exoplanets, where such tests are rarely performed.

In Section~\ref{sec:problem}, we present our formalism. The optimal criterion is searched as a maximum utility action in Section~\ref{sec:utility_some}, and a maximum number of true detections under constraints on the number of false ones in Section~\ref{sec:constraint}. In Section~\ref{sec:discussion}, we present what the optimal procedure is, we show an example of application, and we discuss the relationship of our work with the Bayesian approaches to false discovery rates, as well as ``label-switching'' problems arising in particular in mixture models. We address the problem of model criticism and calibration in Section~\ref{sec:mopen}, and conclude in Section~\ref{sec:conclusion}.

      
   \section{Formalism}
   \label{sec:problem}
   \label{sec:fip}

Let us consider a dataset $y \in \mathcal{Y}$, potentially exhibiting an unkown number of components $n$, such that the component $i=1,..,n$ is parametrized by vector $\theta_i \in T$. The parameters of the nuisance signals are denoted by $\eta \in H$. Our results are valid if $\mathcal{Y}$, $T$ and $H$ are measurable spaces, and in most practical cases their elements are real valued vectors.

To simplify notations, we denote by $p(x)$ the distribution of the random variable $x$.
We assume that we have a likelihood function $p(y\mid (\theta_i)_{i=1...n}, \eta)$ and a proper prior distribution $p((\theta_i)_{i=1...n}, \eta)$. The number of components $n$ is a free parameter, 
we assume that $0\leqslant n \leqslant n_{\mathrm{max}}$ where $n_{\mathrm{max}}$ can be a positive integer or $+\infty$. Writing $\theta = (\theta_i)_{i=1..n}$, $\theta$ belongs to $T \amalg T^2 \amalg... \amalg T^{n_{\mathrm{max}}}$, where $\amalg$ is the disjoint union.

 

In the context of exoplanets, $n$ is the number of planets and $\theta_i$ are the parameters of planet $i$. The parameters $\eta$ is a vector that includes instrumental offsets, potentially polynomial trends, hyperparameters of a stochastic process describing the noise, etc. The likelihood describes the noise model, and the priors are either reference ones or aimed at representing the demographics of planets.


\subsection{Detections}
   \label{sec:mathprob}
  To express that we want a certain accuracy on the component parameters, we define a detection claim as follows. 
  \begin{definition}[Detection claim] We first choose $\mathcal{T}$, a set of measurable subspaces of $T \amalg T^2 \amalg... \amalg T^{n_{\mathrm{max}}}$.  Given $ (\Theta_1,..., \Theta_n) \in \mathcal{T}$, a detection claim is denoted by $a(\Theta_1,...,\Theta_n)$ and defined as the event ``There are exactly $n$ components, one with parameters in $\Theta_1$,  one with parameters in $\Theta_2$...,  one with parameters in $\Theta_n$''. 
 \label{def:detec}
 \end{definition}
 One of the possibilities, if $T$ is a metric space, is to choose $\mathcal{T}$ such that the $\Theta_i$ are closed balls of fixed radius. This conveys the idea that a certain resolution on the parameters is desired. 
 
 A detection claim is completely correct if there are $n$ components truly present in the data, whose parameters $(\theta_i)_{i=1...n}$ are such that there exists a permutation of the indices $\sigma$ with  $\theta_{\sigma(i)} \in \Theta_i$ for all $i = 1,...,n$. The permutation simply expresses that the order of the $\Theta_i$s does not matter. There can be false and missed detections, which we define as follows.
 \begin{definition}[False detections] 
If for a detection claim of $n$ components in $(\Theta_i)_{i=1..,n}$, $k$ of the $\Theta_i$s  do not contain true parameters, we count $k$ false detections.
 \label{def:falsepos}
 \end{definition}
     \begin{definition}[Missed detections] 
If there are $k$ signals truly in the data whose parameters are not contained in any of the $\Theta_i$s, we count $k$ missed detections.
 \label{def:misseddec2}
 \end{definition}
In a given detection claim, there can be simultaneously several false detections and missed detections. In Fig.~\ref{fig:schematic}, we show an example of a detection claim with one false detection and
two missed detections.
\begin{figure}
      \centering
      \includegraphics[width=\linewidth]{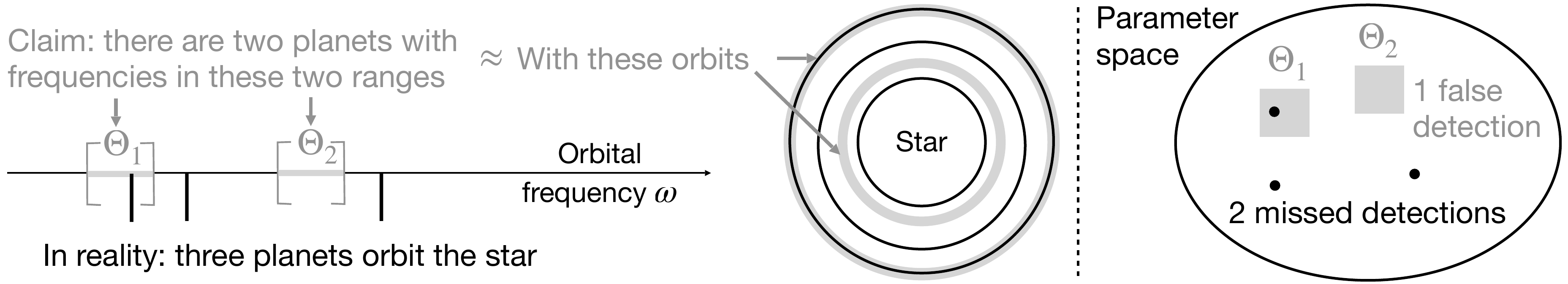}
      \caption{Illustration of a detection claim in the case where $\Theta_i$s are orbital frequency ranges. For a particular dataset of observations of a given star, the claim is that there are two planets with frequency in certain ranges (in gray), which correspond to certain orbits. In fact, there are three planets. There is one true detection, one false, and two missed. On the right, we show schematically the parameter space: we claim detections in the two gray areas. The true parameters of the three components present, represented as black dots, are not all detected. }
      \label{fig:schematic}
  \end{figure}
  
As a remark, supposing there are $m$ different types of parameters, we can define $T= T^1\amalg T^2...\amalg T^m $. The likelihood would write $p[y \mid(\theta^1_i)_{i=1,..n_1},...,(\theta^m_i)_{i=1,..n_m}, \eta ]$  where $\theta^j_i \in T^j$.  For instance, we might want to simultaneously search for planets, moons, and comets around other stars, with parameters in $T^1, T^2, T^3$, respectively.

 
\subsection{The false inclusion probability}
We will see that the choice of $\Theta_i$s minimizing the number of errors involves the posterior probability of the event ``there is one component with parameters in $\Theta_1 \subset T$'', we denote this probability by $\mathrm{TIP}_{\Theta_1}$. It can be expressed as
  \begin{align}
\mathrm{TIP}_{\Theta_1} & :=  \sum\limits_{k\geqslant 1} p(k\mid y) \mathrm{TIP}^k_{\Theta_1}  , \label{eq:pa}  
   \end{align}
where $\mathrm{TIP}^k_{\Theta_1}$ is the probability of the event ``There is one component with parameters in $\Theta_1\subset T$, knowing there are $n$ components'',
     \begin{align}
\mathrm{TIP}^k_{\Theta_1}&:= \int_{\exists i, \theta_i \in\Theta_1 } p(\theta_1,..,\theta_k, \eta \mid y,k) \dd \theta_1...\dd \theta_k \dd \eta \label{eq:ik}
 \end{align}
where $p(k \mid y)$ is the posterior probability of having $k$ components. It can be expanded as
\begin{align}
    p(k \mid y) = \frac{p(y\mid k)p(k)}{\sum_{i=0}^{n_{\mathrm{max}}} p(y\mid i)p(i)}  ,
    \label{eq:pk}
\end{align}
   where $p(k)$ is the prior probability to have $k$ components, and 
 \begin{align}
	 	\label{eq:evidence}
	 	p(y \mid  k) =\int  p( y \mid k, \theta_1,..,\theta_k, \eta)p(\theta_1,..,\theta_k, \eta) \dd \theta_1...\dd \theta_k \dd \eta 
\end{align}
is the Bayesian evidence, integrated on all possible combinations of $k$ components. We refer to $\mathrm{TIP}_{\Theta_1}$ as true inclusion probability (TIP) to have a component in $\Theta_1$.  The TIP can be seen as an extension of the posterior inclusion probability (PIP), applicable to linear models~\citep{barbieriberger2004}, to the general model considered here.

 %
 Some formulae are more conveniently written with the probability not to have a component in $\Theta_1$. Following~\cite{hara2021} we call it the false inclusion probability (FIP), defined as
\begin{align}
\mathrm{FIP}_{\Theta_1} := 1 - \mathrm{TIP}_{\Theta_1}.
\label{eq:fipdef}
 \end{align}
 
\subsection{Practical computations} The quantity $\mathrm{TIP}_{\Theta_1}$ in Eq.~\eqref{eq:pa} is decomposed in terms $p(k \mid y)$ and $\mathrm{TIP}_{\Theta_1}^k$ for $k=0...n_\mathrm{max}$. Those are  easily computed as a by-product of the output of algorithms estimating Bayesian evidences of the models with a fixed number $k$ of components (Eq.~\eqref{eq:evidence}). Below, we outline the different steps of the calculation and provide a detailed example in a Python notebook \footnote{\label{note1}Available in supplementary material~\citep{hara2023_supp_code}, and  \url{https://github.com/nathanchara/FIP}.}. 

Let us suppose that for a fixed $k$, there are $N$ posterior samples of the distribution $p(\theta_1,..,\theta_k, \eta \mid y,k)$, $\theta^i = (\theta_1^i,..,\theta_k^i)$, $i=1..N$, reliably exploring the parameter space. The quantity $ \mathrm{TIP}^k_{\Theta_1}$ in Eq.~\eqref{eq:ik} can simply be estimated as the number of indices $i$ such that for some $j =1,..,k$, $\theta_j^i \in \Theta_1$, divided by the total number of samples $N$.  Alternatively,  $\mathrm{TIP}^k_{\Theta_1}$ can be estimated with a nested sampler~\citep{skilling2006}. For a model with $k$ components, nested samplers provide a collection of $N$ samples $\theta^i = (\theta_1^i,..,\theta_k^i)$, $i=1,..N$, associated to a weight $w_i$ and a likelihood $L_i$. Denoting by $p_i = w_i L_i /\sum_j w_j L_j $, $ \mathrm{TIP}^k_{\Theta_1}$ in Eq.~\eqref{eq:ik} is simply the sum of $p_i$ taken over the indices of samples with $\theta^i $ such that $\theta^i_j \in \Theta_1$ for some $j$.  

The computations of Bayesian evidences themselves are difficult in the case of exoplanets~\citep{nelson2020}. We suggest a way to test their reliability in the context of FIP calculations in~\cite{hara2022}, implemented in the example notebook\textsuperscript{\ref{note1}}.

\subsection{Example}
\label{sec:fipperio}

To illustrate the notions given above, let us consider the radial velocity time-series of HD 21693~\citep{udry2019}, taken with the HARPS instrument and processed with the YARARA pipeline~\citep{cretignier2021}. The RV technique consists in acquiring the spectrum of a given star at different times. Thanks to the Doppler effect, by measuring the frequency shift between these spectra, astronomers can determine the variation of the radial velocity (RV) of the star: its velocity projected onto the line of sight. If an orbiting planet is present, it causes a reflex motion of its host star, and thus periodic RV variations. If several planets orbit the star, their combined effect is well approximated by the sum of their individual ones.
Denoting by $y$ the RV time series analysed, our likelihood is defined by the model of the measurement at time $t$,
\begin{eqnarray}
    y(t) &=& c_0  + \sum\limits_{j=1}^{n} f(t,\theta_i)  + \epsilon_t \label{eq:nominal_model} \\ 
     \epsilon_t &\sim & \mathcal{N}(0,\sigma_t^2 + \sigma_J^2),  \label{eq:nominal_model2}
\end{eqnarray}
where $\theta_i$ is the vector of the orbital parameters of planet $i$, $f$ is a periodic function as defined in Eq. 1 of~\cite{wrighthoward2009}, and $\sigma_t$ is the nominal error bar on the measurement at time $t$. The free parameters are $n$, $c_0,\sigma_J,(\theta_i)_{i=1..n}$.
We assume that there are at most $n_{max}=3$ planets. We fix uninformative priors, which we do not specify for the sake of brevity, as they do not matter here. The posterior distributions and Bayesian evidences are computed with the nested sampling algorithm  \textsc{POLYCHORD},~\citep{handley2015b,handley2015}.

As in~\cite{hara2021, toulisbean2021}, we consider that a planet has been correctly detected if the period of the claimed planet is close to a true period with a certain accuracy. In the formalism of section~\ref{sec:framework}, we choose the sets $\Theta_i$s as frequency intervals of fixed length, chosen from a grid.  The interval $k$ of the grid is defined as $[\omega_k - \Delta \omega/2, \omega_k + \Delta \omega/2]$, where $ \Delta \omega = 2\pi/T_\mathrm{obs}$, $T_\mathrm{obs}$ is the total observation time span, and  $ \omega_k = k\Delta \omega / N_\mathrm{oversampling}$. We take $N_\mathrm{oversampling} = 5$.  
For each interval in the grid, we compute the marginal probability to have a planet whose frequency lies in the interval, and its FIP as defined in Eq.~\eqref{eq:fipdef}\footnote{It is in fact more computationally efficient to loop over the samples and determine in which frequency interval they lie, see the Python example in supplementary material~\citep{hara2023_supp_code}, available also at \url{https://github.com/nathanchara/FIP}.}. The value of $-\log_{10}$FIP as a function of frequency is represented in Fig.~\ref{fig:fipperio}, and is called a FIP periodogram.

In the following sections, we show that the optimal detection procedure simply consists in claiming the detection of all signals with FIP lower than a certain threshold (see section~\ref{sec:procedure}). 
Here, for any reasonable FIP threshold, 0.01\%-10\%,  two planets are detected. Indeed, 
the probability that there are no planets with frequencies   $2\pi/53.8\pm \Delta \omega$ rad/day and $2\pi/ 22.67\pm \Delta \omega$ rad/day period is supported with FIP less than $10^{-6}$. The interval with the third highest probability occurs at 4118 days, it has a FIP greater than 90\%, and likely stems from low frequency noise. 



\begin{figure}
    \centering
    \includegraphics[width=13cm]{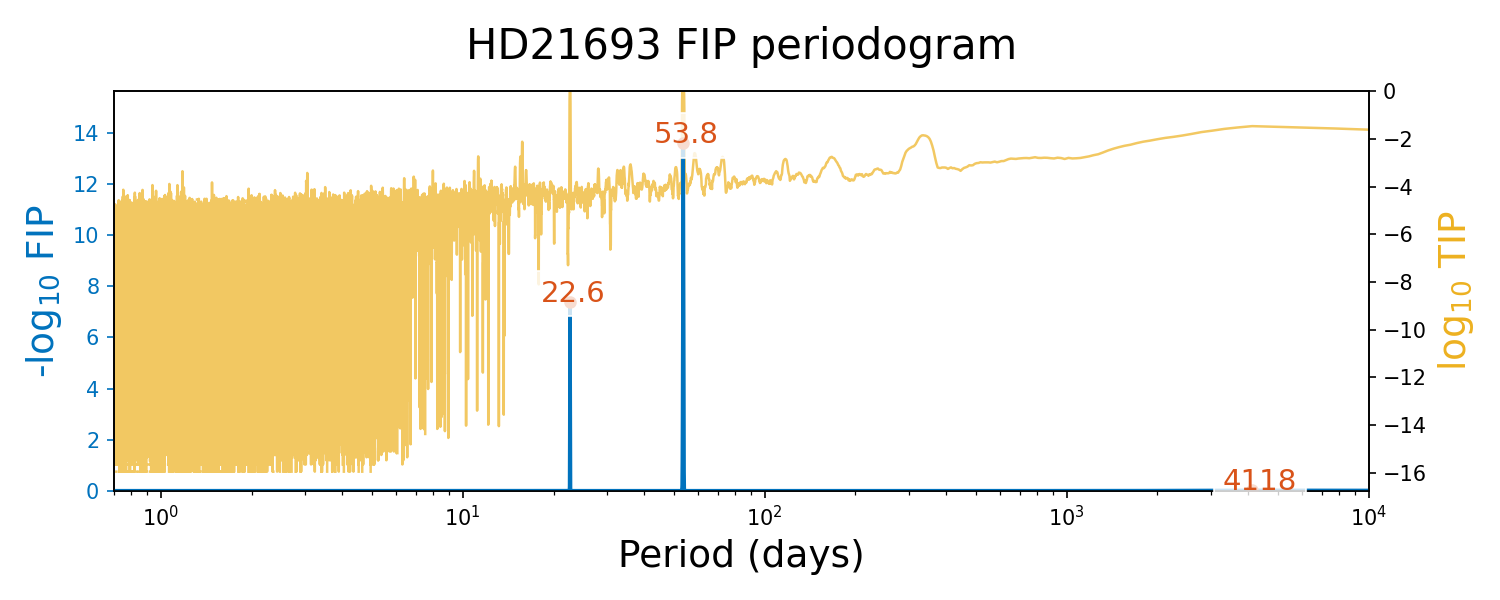}
    \caption{FIP periodogram of the radial velocities measured on HD 21693~\citep{udry2019}: we represent $-\log_{10}\mathrm{FIP}$ of having a planet with frequency in interval $[\omega_k - \Delta \omega/2, \omega_k + \Delta \omega/2]$ as a function of $\omega_k$ for a maximum of 2 planets (scale on the left $y$ axis)). The scale on the right $y$ axis represents TIP = 1- FIP. The quantities $-\log_{10}\mathrm{FIP}$ and $\log_{10}\mathrm{TIP}$ represent differently the same information: $-\log_{10}\mathrm{FIP}$ allows to distinguish between very significant and moderately significant signals, $\log_{10}\mathrm{TIP}$ allows to see hints of signal.} 
    \label{fig:fipperio}
\end{figure}


       \section{Maximum utility}
        \label{sec:utility_some}
     In the present work, we consider a dataset $y$ that can support the detection of components, described by a likelihood and a prior distribution (see Section~\ref{sec:mathprob}).  A detection claim, as in Definition~\ref{def:detec}, will be issued: we will output a number of components and locations of their parameter. We want this detection claim to be optimal in the sense that it minimizes both false and missed detections. 
     We first formalize this problem in the maximum utility framework~\citep{vonneumann1947}. 
    
    The maximum utility framework requires us to specify a utility function over a space of potential actions.  
    Here, the space of actions is 
     \begin{align*}\mathcal{A} = \{a(\Theta_1,...,\Theta_n), (\Theta_i)_{i=1...n} \in \mathcal{T} \}\end{align*}  where $a(\Theta_1,...,\Theta_n)$ is given in Definition~\ref{def:detec}. Choosing an action here is exactly equivalent to choosing where we claim the parameters of the components are: $(\Theta_i)_{i=1...n} \in \mathcal{T}$. 
     

    The utility function
    $U\{a,(\theta,\eta)\}$ is a non-negative real number
     quantifying the gain of taking action $a$ while the true value of the parameters of the components and the nuisance parameters are $\theta = (\theta_i)_{i=1,..,n}$ and $\eta$, respectively.
  For fixed parameters $(\theta,\eta)$, the higher the gain of taking action $a$ is,  the higher $U\{a,(\theta,\eta) \}$ should be.


 

    In practice, the true values of $(\theta,\eta)$ are unknown, the only information we have about it are their posterior distribution.
   The maximum utility action $a^\star$, if it exists, is defined as maximizing the expectancy of the utility taken on the posterior distribution of the pararameters $(\theta,\eta)$,
      \begin{align} \label{eq:fullprob}
        a^\star = \argmax\limits_{a \in \mathcal{A}} E_{\theta,\eta}\left[ U\left\{a, (\theta,\eta)\right\} \right]   \ 
       \end{align}
       where
       \begin{align}
    E_{\theta,\eta}\left[ U\left\{a, (\theta,\eta)\right\} \right] = \int U\left\{a, (\theta,\eta)\right\} p(\theta, \eta \mid y)\dd\theta  \dd\eta  .
       \end{align}

        \vspace{-0.3cm}
           \subsection{Defining the utility function}
   \label{sec:framework}
 We adopt the definition of false and missed detections given in Section~\ref{sec:framework}, and we assign the value 0 to the utility if all the detections are correct. Each false detection has  a cost $\alpha>0$, and a missed detection has a cost $\beta\geqslant 0$.  Denoting by $E[\mathrm{FD}]$ and $E[\mathrm{MD}]$ the expected number of false detections and missed detections when choosing $(\Theta_i)_{i=1...n}$, the expected utility could be defined as $- \alpha E[\mathrm{FD}] - \beta E[\mathrm{MD}]$. However, since we're interested in the maximum, it is unnecessary to keep both $\alpha>0$ and $\beta$ as parameters, and we define the expected utility as $- E[\mathrm{FD}] - \gamma E[\mathrm{MD}]$ where $\gamma = \beta/\alpha$.    
In Appendix~\ref{app:b}, we show that the expected utility of choosing $\Theta_1,...,\Theta_n$ is
\begin{align}
 E_{\theta,\eta}\left[ U\left\{a(\Theta_1, ..., \Theta_n), (\theta,\eta)\right\} \right]   =  -n + (1 + \gamma) \sum\limits_{j=1}^n j I_{A_j}(\Theta_1, ..., \Theta_n) -  \gamma  \sum\limits_{k=1}^{n_{max}} k p(k\mid y), \label{eq:ujplagamma_othercriterion}
\end{align}
where $p(k\mid y)$ is defined in Eq.~\eqref{eq:pk} and $I_{A_j}(\Theta_1, ..., \Theta_n)$ is the probability that exactly $j$ detections are correct when choosing $\Theta_1,...,\Theta_n$, or in other words, the probability to have exactly $j$ true detections.  

The expected utility is a function of  $(\Theta_i)_{i=1..n}$, and we want to maximize it. To do so, we first maximize it with a fixed number of components $n$, then find the optimal $n$. Maximizing the expected utility (Eq.~\eqref{eq:ujplagamma_othercriterion}) for a fixed $n$ is equivalent to finding 
   $ (\Theta_i^n)_{i=1..n} $ maximizing $\sum_{j=1}^n j I_{A_j}(\Theta_1,...,\Theta_n),
    \label{eq:max_util_n}$
 provided they exist. \textit{A priori}, for $n' \neq n$ the optimal $(\Theta_i^n)_{i=1..n}$ and $(\Theta_i^{n'})_{i=1..n'}$ do not necessarily have something in common.  As we shall see in the following, if the $\Theta_i$s are ``separable'' in some sense, then the optimal spaces $(\Theta_i^n)_{i=1..n}$ can be found iteratively as $n$ increases.
 

%
    

\subsection{Components with parameters in disjoint regions}


A first simplification occurs if we restrict the search to $\Theta_i,i=1...n$ that are pairwise disjoint. One possible justification is that the physics forbids two components to occupy the same $\Theta_i$. 
\begin{lemma}
 Let us consider $\Theta_1 \subset T,...,\Theta_n \subset T$    , $ \forall i_1, i_2 = 1...n, i_1 \neq i_2 $, $\Theta_{i_1} \cap \Theta_{i_2} = \emptyset$. Then with $\mathrm{TIP}_{\Theta_i} $ as defined in Eq.~\eqref{eq:pa}, and  $I_{A_j}(\Theta_1,...,\Theta_n) $ as used in Eq.~\eqref{eq:ujplagamma_othercriterion} 
\begin{align}
 \sum\limits_{j=1}^n j I_{A_j}(\Theta_1,...,\Theta_n)   = \sum\limits_{i=1}^n \mathrm{TIP}_{\Theta_i}
 \label{eq:aj_eq_thetai}
\end{align}
\label{lem:separ}
\end{lemma}
The proof, given in Appendix~\ref{app:c}, simply uses a decomposition of the terms of the left-hand sum. 
The quantity $\mathrm{TIP}_{\Theta_i}$ can be computed as explained in Section~\ref{sec:fip}. 
If Eq.~\eqref{eq:aj_eq_thetai} is verified,  adopting the definition of FIP of Eq.~\eqref{eq:fipdef}, 
maximising the expected utility, as defined in Eq.~\eqref{eq:ujplagamma_othercriterion}, for a fixed number of components $n$ comes down to solving 
\begin{align}
(\Theta_i^n)_{i=1..n} = \argmin\limits_{\Theta_1,... \Theta_n ,  \Theta_i \cap \Theta_j = \emptyset} 
\sum\limits_{i=1}^n    \mathrm{FIP}_{\Theta_i} = \argmax\limits_{\Theta_1,...,\Theta_n ,  \Theta_i \cap \Theta_j = \emptyset}  \sum\limits_{i=1}^n    \mathrm{TIP}_{\Theta_i} .
\label{eq:maxn}
\tag{$P_n$}
\end{align}
We just want to find $(\Theta_i^n)_{i=1..n}$ maximizing the probability that they ``host'' a component.
Note that the equality of Eq.~\eqref{eq:aj_eq_thetai} is not true if the $\Theta_i$s are allowed to have non empty intersection. This is easily seen for $n=2$, where  $I_{A_1}^1 = I_{\Theta_1 \cup \Theta_2}^1$ which might not be equal to $ \mathrm{TIP}_{\Theta_1}^1 +  \mathrm{TIP}_{\Theta_2}^1$ if $\Theta_1$ and $\Theta_2$ are not disjoint. 
In Appendix~\ref{app:d} we exhibit some conditions guaranteeing the existence of the solution to~\eqref{eq:maxn}. In the next section, we try to find it. 


%


\subsection{Searching for the optimum}
\label{sec:searching}

To solve $\eqref{eq:maxn}$ for $n=1$, we simply need to find the region of the parameter space ${\Theta^1_1}$ with maximum $\mathrm{TIP}_{\Theta^1_1}$. It is tempting to build the solution for higher $n$ from there. Supposing that~\eqref{eq:maxn} has a solution $(\Theta_i^n)_{i=1...n}$, a natural candidate solution for $(P_{n+1})$ is to append the space with maximum true inclusion probability outside the $\Theta_i^n$s, which we denote by
\begin{align}
    \Theta_{n+1}^\star = \argmax\limits_{\Theta_{n+1} \in T \backslash \cup_{i=1}^n \Theta_i^n } \mathrm{TIP}_{\Theta_{n+1}} = \argmin\limits_{\Theta_{n+1} \in T \backslash \cup_{i=1}^n \Theta_i^n } \mathrm{FIP}_{\Theta_{n+1}}.
    \label{eq:thetastar}
\end{align}
The following result, proven in Appendix~\ref{app:e},  shows that this simple procedure is not always optimal. As expressed in the following lemma, the solution to  $(P_{n+1})$ either consists in appending $\Theta_{n+1}^\star$, or in arranging $n+1$ regions all intersecting at least one of the $(\Theta_{i}^n)_{i=1,..,n}$.
\begin{lemma}
Suppose that~\eqref{eq:maxn} has a solution $(\Theta_1^n,...,\Theta_n^n) \in \mathcal{T}$. 
Then the solution to $(P_{n+1})$ is either $(\Theta_1^n,..., \Theta_n^n, \Theta_{n+1}^\star )$ or such that $\forall i = 1...n+1$, $\exists j =  1...n$, $\Theta_{i}^{n+1} \cap  \Theta_{j}^{n} \neq \emptyset$.
\label{lem:find}
\end{lemma}

In the example of Section~\ref{sec:fipperio}, we wanted to determine the presence of a planet in a frequency interval of size $\Delta \omega =2\pi/T_{\mathrm{obs}}$, the timespan of the observations $T_{\mathrm{obs}}$ dictates the frequency resolution. Had we chosen a much smaller $\Delta \omega$, when searching for the solution to~\eqref{eq:maxn}, as $n$ increases, the new interval minimizing~\eqref{eq:maxn} would be ``glued'' together with the $n-1$ present to form a wider interval, until enough probability mass around the TIP modes is encompassed. In this example and elsewhere, such iterative ``gluing'' essentially means that the $\Theta_{i}$s do not capture the resolution provided by the likelihood, they are too small. To circumvent this problem, the size of $\Theta_i$s can be defined as the typical size of posterior modes.

Lemma~\ref{lem:find} motivates the definition of separability, simply expressing the condition under which it can be ensured that the solution to $(P_{n+1})$ is to append $\Theta_{n+1}^\star$ (Eq.~\eqref{eq:thetastar}).
\begin{definition}[Separability] We say that a dataset $y$, with likelihood $p(y\mid \theta,\eta)$ and prior $p(\theta,\eta)$ verifies component separability of order $n \geqslant 2$ if (i) the solution to $(P_{n-1})$, $\Theta_1^{n-1},...,\Theta_{n-1}^{n-1}$ are pairwise disjoint and $(ii)$ the solution to~\eqref{eq:maxn} is $(\Theta_1^{n-1},..., \Theta_n^{n-1}, \Theta_{n}^\star )$ as defined in Eq.~\eqref{eq:thetastar}. 
\label{def:separability}
\end{definition}
If $T$ is a metric space, separability is ensured as long as in the vicinity of the $\Theta_i$s, the probability of having a signal is sufficiently low. We denote by $B(\theta,R)$ a closed ball of centre $\theta$ and radius $R$. In Appendix~\ref{app:e} we prove the following result.
\begin{lemma}
 Let us suppose that $\mathcal{T}$ is the set of 0 to $n_{\mathrm{max}}$ disjoint balls of radius $R$ and centres $(\theta_i)_{i=1...n}$. Denoting by  $\Theta^c = \cup_{i=1}^n B(\theta_i, 3R) \backslash  \cup_{i=1}^n \Theta_i$, if $\mathrm{TIP}_{\Theta^c} < \mathrm{TIP}_{\Theta^\star_{n+1}} $, then the components are separable of order $n+1$.
 \label{lem:balls}
\end{lemma}

   \section{Minimum missed detections under constraints}
   
\label{sec:constraint}
 
In Section~\ref{sec:utility_some}, we minimized the number of missed and false detections with a cost function. Alternatively, we can consider the following problem: among the possible choices of $n, (\Theta_1,..\Theta_n)$ such that the expected number of false detections is lower than a certain $x \in [0, n_{\mathrm{max}}]$, which choice minimizes the number of missed detections? 
This is equivalent to maximizing the number of true detections subject to a constraint on the expected number of false ones. In Appendix~\ref{app:b}, we show that the constrained problem is
\begin{align}
\argmax\limits_{(n,\Theta_1,..\Theta_n)}  \sum\limits_{j=1}^n j I_{A_j}(\Theta_1,...,\Theta_n) 
 \text{  subject to }  n- \sum\limits_{j=1}^n j I_{A_j}(\Theta_1,...,\Theta_n)  \leqslant x.
 \label{eq:constraints_othercriterion}
\end{align} 
Maximizing \eqref{eq:ujplagamma_othercriterion} can be seen as a ``Lagrange multiplier'' version of Eq.~\eqref{eq:constraints_othercriterion}.
A natural question is whether the solutions of~\eqref{eq:constraints_othercriterion} obtained as $x$ skims $[0,+\infty)$ also maximize Eq~\eqref{eq:ujplagamma_othercriterion}, when $\gamma $ skims $[0,+\infty)$. The following theorem, proven in Appendix~\ref{app:f} shows that it is true if we have component separability as in Definition~\ref{def:separability}. 
\begin{theorem}
Let us consider a dataset $y$ and suppose that it verifies component separability at all orders, then there exists an increasing function $\gamma(x) \geqslant 0$  such that the solution of Eq.~\eqref{eq:constraints_othercriterion} and the solution maximizing the utility in Eq.~\eqref{eq:ujplagamma_othercriterion} are identical. 
\label{theorem}
\end{theorem}




\section{Discussion}
\label{sec:discussion}
\vspace{-0.3cm}

Our initial aim was to find an optimal procedure to determine which parametric components are in the data, in the framework described in section~\ref{sec:problem}. 
 We now discuss the final optimal procedure, when it is applicable, how it performs compared to other criteria and how it relates to other approaches. 
        
\subsection{Procedure}
\label{sec:procedure}
We considered two formulations of optimality: the maximum utility decision, and minimum expected number of missed detections under a constraint on the expected number of false detections.
 We have seen in section~\ref{sec:constraint} that the two are equivalent, provided a technical condition is satisfied (see Definition~\ref{def:separability}). 
If it is, 
 suppose we have found the solution $(\Theta_i^n)_{i=1,..,n}$ solving~\eqref{eq:maxn}, that is maximizing the utility with $n$ components. The maximum utility with $n+1$ components is obtained by appending $\Theta_{n+1}^{\star}$, the space with minimum FIP outside the $(\Theta_i^n)_{i=1,..,n}$, as defined in Eq.~\eqref{eq:thetastar}. 
The detection of component $n+1$ is claimed if the expected utility (Eq.~\eqref{eq:ujplagamma_othercriterion}) with $n+1$ component is greater than the one with $n$ components, 
\begin{align}
 \mathrm{TIP}_{\Theta_{n+1}^{\star}} \geqslant \frac{1}{\gamma+1}   \iff  \;\mathrm{FIP}_{\Theta_{n+1}^{\star}} \leqslant \frac{\gamma}{\gamma+1},  
    \label{eq:criterionFIP_fixed}
\end{align}
which expresses the intuitive idea of a bet ``$\gamma$ to one''. If the cost of missing a detection is  $\gamma$  when the cost of a false one is 1, then the FIP of the new signal should be less than $\gamma/(\gamma+1)$. Typically, we would choose $\gamma$ small to penalize false detections more than missed ones. This procedure is consistent with a ``Dutch book'' principle~\citep{ramsey1926, definetti1937}, the best significance metric is the probability that there is a component with certain parameters.


In the maximum utility case, one must know the maximum number of components possible $n_\mathrm{max}$, and be able to compute the Bayesian evidence up to that number, which might be computationally heavy or even intractable. To avoid computing Bayesian evidences for high dimensional models, we suggest to compute the FIPs each time $n_\mathrm{max}$ is incremented, until the FIPs of the candidate signals do not vary. This procedure is described in detail in the case of exoplanets in~\cite{hara2021}\footnote{See also the Python notebook example in supplementary material~\citep{hara2023_supp_code}, also available at \url{https://github.com/nathanchara/FIP}.}. In~\cite{hara2021}, the definition of false negatives is slightly different than the one adopted here. In Appendix~\ref{app:g}, we show that with that definition, the optimal criterion is
\begin{align}
    \mathrm{FIP}_{\Theta_{n+1}^{\star}} \leqslant \gamma p(k\geqslant n+1\mid y) .
    \label{eq:criterionFIP_body}
\end{align} 
where $ p(k\geqslant n+1\mid y) $ is the probability that there are $n+1$ components or more. 

\subsection{Applicability: general case}

Our formalism requires to choose the space $\mathcal{T}$, that is the possible choices of $\Theta_i$s (see definition~\ref{def:detec}). If $T$ is a metric space, we recommend to choose $\Theta_i$s as balls of radius $R$. 
If the prior probability to have components with parameters $\theta_1$ and $\theta_2$ with $| \theta_1- \theta_2| <2R$ is vanishing, $\Theta_i$s can be assumed pairwise disjoint, and it is sufficient to solve~\eqref{eq:maxn}.

Another important question is whether the separability condition~\ref{def:separability} applies, in which case one can simply follow the procedure of Section~\ref{sec:procedure}. 
Thanks to Lemma~\ref{lem:balls}, if the $\Theta_i$s are chosen as balls of fixed radius it is easy to verify that the signal is indeed separable at order $n+1$, and ensure that appending $\Theta_{n+1}^\star$ to the solution of~\eqref{eq:maxn} solves $(P_{n+1})$.
In practice, the value of $R$ can always be chosen small enough to verify the separability condition, unless there can be components with exactly the same parameters. If such a situation happens, it likely means that the $\Theta_i$ should be redefined. For instance, parametric descriptions of components typically include an amplitude, and the superposition of two identical components has simply twice the original amplitude. If simplifications are not possible, one can maximise the expected utility in the general case (see Eq.~\eqref{eq:ujplagamma_othercriterion}). 
 
In our framework, the candidate $\Theta_i$ spaces in $\mathcal{T}$  might have different sizes. It might be useful, for instance, to adopt $\Theta_i$s which adapt to the size of posterior local maxima. If the size of $\Theta_i$s is allowed to vary, the maximisation might favour larger $\Theta_i$ spaces. To prevent this situation, one can add a penalisation term to Eq.~\eqref{eq:ujplagamma_othercriterion} for the size of the $\Theta_i$, for instance, by choosing a certain measure $\mu$ on $T$ and adding a term $-\sum_i \mu(\Theta_i)$ to  Eqs.~\eqref{eq:ujplagamma_othercriterion}. We leave this latter approach for future work.

\subsection{Application to exoplanets}

In the case considered in Section~\ref{sec:example}, we want to localise the period of exoplanets. Lemma~\ref{lem:balls} means that if the grid of intervals has width $\Delta \omega = 2R$, and the prior excludes two planets being closer than $2\Delta \omega$, the separability condition is verified. This is true in general, because planets with very close orbits are generally dynamically unstable: there would be a collision, or one of the planets would be ejected in short time relative to the age of the system.

There are two cases where separability is a concern, but with little practical impact. The first case is the potential presence of co-orbital planets, sharing the same orbit \citep[e.g][]{gascheau1843}. To address that case, it is possible to either adopt the general formalism of section~\ref{sec:framework}, or to further specify the $\Theta_i$ spaces to ranges of orbital frequencies and phase, such that two planets with parameters in the same $\Theta_i$ would be unstable, thus excluded by the prior. 
So far, no co-orbital planets have been detected, and it has been shown that they would have a specific signature in the RV~\citep{leleu2015}.   

Second, as $\omega$ decreases, $\omega$ and $\omega\pm2\Delta \omega$ correspond to increasingly different periods.  
 Following~\cite{deck2013} Eq. 49, two planets might be unstable if they are not in coorbital resonance and their period ratio is less than $1+2.2\epsilon^{2/7}$ where $\epsilon$ is the ratio of the sum of planetary masses to the stellar mass. Defining $T_{\mathrm{obs}}$ as the observation timespan, and $\Delta\omega =2\pi/T_{\mathrm{obs}}$, in principle there can be two Neptune-mass planets with frequencies less than $2\Delta \omega$ away from one another if $\omega\leqslant 15 \Delta\omega$. In principle, separability might be a concern in the first 15 peaks of the periodogram.

The fact that two planets with frequencies close enough appear as a single planetary signal concerns not only the FIP, but all exoplanet detection methods. It is due to the intrinsic frequency resolution of a dataset, set by the sampling and in particular its baseline.  This is a well known fact, and any detection of long period planets is subject to caution. 
 In the next section, we apply the FIP to 1000 datasets simulated without restrictions on the proximity of two periods, such that the separability condition is not always verified. The optimal procedure still outperforms other methods. 
 
In conclusion, the separability condition is not a concern for the application of exoplanets, but as for any other detection metric, at long periods the frequency resolution of a dataset limits the ability to ascertain that a given signal is due to a single planet.




\subsection{Example}
\label{sec:example}

To show how the selection criterion~\eqref{eq:criterionFIP_fixed} performs compared to other significance metrics, we apply it to the detection of sinusoidal signals in unevenly sampled data, where a certain precision on the frequency is desired. We use the simulations of~\cite{hara2021} that emulate a search for exoplanets with radial velocity. The data are thus in velocity units (m/s).  

We simulate 1000 time-series with 80 time stamps $t$, taken as those of first 80 HARPS measurements of HD 69830~\citep{lovis2006}, which are representative of a typical radial velocity sampling. We inject a signal of the form $C + \sum_{i=1}^k A_k\cos2\pi t/P_k + B_k\cos2\pi t/P_k$, with $k=0$ (no signal), $k=1$, or $k=2$. 
The values of elements $k,A,B,C,P$ are sampled from distributions shown in Table~\ref{tab:priorscirc}. We add a Gaussian white noise according to the nominal error bars, which are typically  0.54 $\pm$ 0.24 m/s.
We then generated another set of 1000 systems with a lower S/N. The simulation is made with identical parameters except that an extra correlated Gaussian noise is added. This one has an exponential kernel with a 1 m/s amplitude and a timescale $\tau = 4$ days. We will refer to these two simulations as the high and low signal-to-noise ratio (SNR) simulations. 

 \begin{table}
 \centering
        \caption{Priors used to generate and analyse the 1000 systems with circular orbits.}
        \begin{tabular}{p{1.8cm} p{4.6cm} p{4cm}}
        Parameter & Prior & Values \\ \hline 
        k & Uniform $[k_\mathrm{min}, k_\mathrm{max}]$ & $k_\mathrm{min} = 0 $, $k_\mathrm{max} = 2$  \\
        A & $\mathcal{N}(0,\sigma_A^2)$ & $\sigma_A = 1.5$ m/s\\
        B & $\mathcal{N}(0,\sigma_B^2)$ & $\sigma_B = 1.5$ m/s\\  
        C & $\mathcal{N}(0,\sigma_C^2)$ & $\sigma_C = 1$ m/s\\            
        P & $\log$-uniform on $[P_\mathrm{min}, P_\mathrm{max}]$ & $P_\mathrm{min} = 1.5 $, $P_\mathrm{max} = 100$
        \end{tabular}
        \label{tab:priorscirc}
 \end{table}

Our goal is to evaluate different detection methods, provided the model used is correct. As a consequence, we analyse the data with the priors and likelihoods used to generate it. The procedures corresponding to Eqs.~\eqref{eq:criterionFIP_fixed} and~\eqref{eq:criterionFIP_body}  are labelled FIP and Max. utility respectively.
We also test the best standard practices of the exoplanet communities. These procedures all proceed in two steps: finding candidate periods and assessing their statistical significance. To find the period, we use either a periodogram as defined in~\citep{delisle2019a}, a FIP periodogram as described in Section~\ref{sec:fipperio}, using all the perior posterior, or a $\ell_1$ periodogram, aiming at retrieving a sparse Fourier spectrum~\citep{hara2017}. The significance is either assessed with a false alarm probability~\citep{delisle2019a}, a Bayes factor defined as $p(y \mid k+1)/p(y \mid k)$ with $p(y \mid k)$ and false alarm probability (FAP) as defined in  Eq.~\eqref{eq:evidence})~\citep{gregory2007}, or taking the number of planets $k$ maximising the posterior number of planets (PNP), $p(k \mid y)$.  The detail of the different methods are given in Appendix~\ref{app:h}.

 %

  To evaluate the performance of the different analysis methods, if the frequency of a detected signal is more than $2\pi /T_\mathrm{obs}$  away from the frequency of an injected signal, it is considered as a false detection. If we miss $n$ planets, we count $n$ missed detections. In other words, we adopt the definition of false and missed detection of Definitions~\ref{def:falsepos} and~\ref{def:misseddec2}, respectively.
 In Fig.~\ref{fig:mistakes_level_white} (a) and (b), we show the  total number of mistakes for the high and low SNR simulations, respectively, made with the criterion~\eqref{eq:criterionFIP_fixed} as a function of the FIP threshold. It appears that the minimum number of mistakes is attained when selecting signals with FIP $\leqslant$ 0.5, when it is more likely than not to have a planet.

  \begin{figure}
\includegraphics[width=\linewidth]{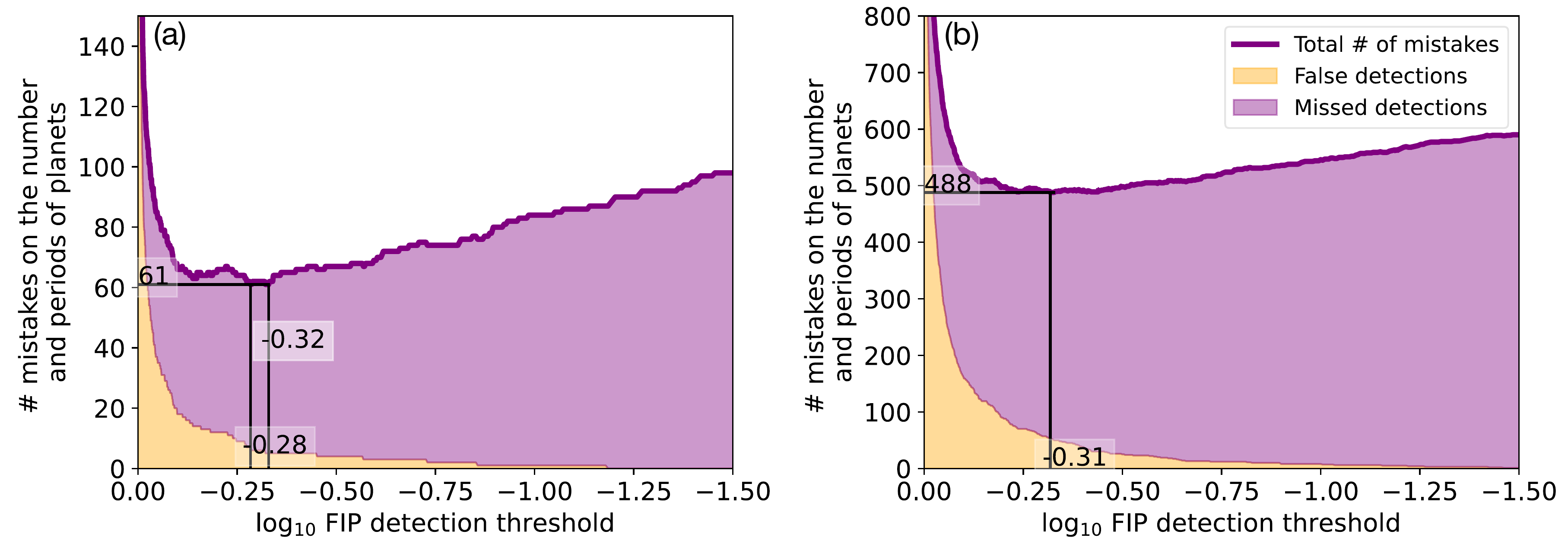}
\caption{Missed + false detections, as a function of the FIP detection threshold as defined in~\eqref{eq:fipdef}. (a) is obtained on the high SNR simulation and (b) the low SNR simulation. In (a) the minimum of mistakes is 61, attained at $\log_{10} \mathrm{FIP} \in [-0.32, -0.28]$ or equivalently $\mathrm{FIP} \in [0.47, 0.52] $. In (b) the minimum of mistakes is 488, attained at $\log_{10} \mathrm{FIP} = -0.31$, or $ \mathrm{FIP} = 0.49$.  }
\label{fig1}
\label{fig:mistakes_level_white}
\end{figure}
  \begin{figure}
\includegraphics[width=\linewidth]{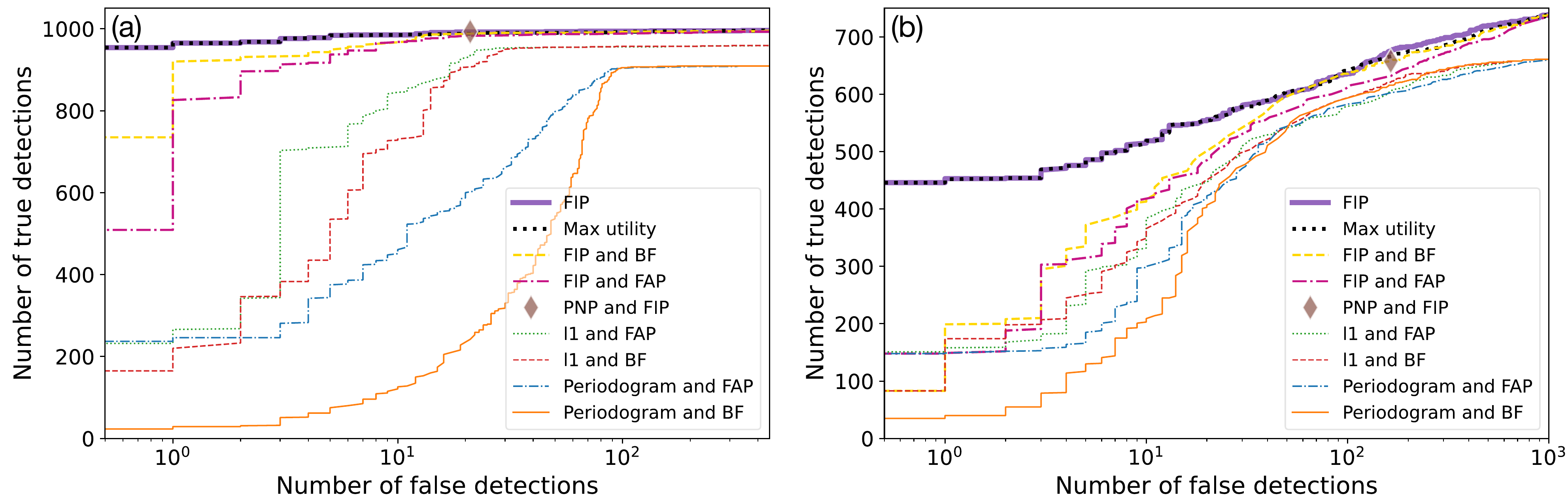}
\caption{True detections as a function of the number of false detections for the different analysis methods as their detection threshold vary, similarly to a ROC curve. (a) and (b) are obtained on the high and low SNR simulations, respectively.   The labels of the curves correspond to the methods used to find the candidate periods and to assess their significance. For instance, Periodogram and Bayes factor means that the candidate period is selected with a periodogram, and its significance assessed with a Bayes factor.  }
\label{fig:mistakes_roc_red}
\end{figure}

  To compare the different detection methods, we compute a curve similar to a Receiving operator characteristic. For each detection method, we vary the detection threshold, which draws the number of true detections as a function of the number of false ones. The optimization problem~\eqref{eq:constraints_othercriterion} explicitly seeks the detection criterion maximizing this curve. 
  In Fig.~\ref{fig:mistakes_roc_red} (a) and (b) we show the curves obtained for the high and low SNR simulations, respectively.
  As expected, the FIP and maximum utility criterion outperform other metrics, notably below $\sim$10 false positives, where it leads to a significantly higher number of true detections. 

We investigated the reasons why the FIP yields more true detections. 
 In the low signal to noise ratio, we compared the false detections obtained when the FIP or FIP periodogram + Bayes factor have each exactly 500 true detections. The corresponding threshold is FIP = 0.1 and Bayes factor = $10^{3.5}$, and there are respectively 7 and 17 false detections. The reason is that the Bayes factor validates the appropriate number of planets, but at the wrong periods. Due to the sampling, we expect that if there is a posterior mode of period at frequency $\omega$, there should be also at  $\omega \pm \omega_i$ where $\omega_i$s are local modes of the spectral window -- the Fourier transform of $\sum_i \delta(t-t_i)$ where $t_i$ are observation times. For evenly sampled data, besides $\omega_0 = 0$, there is a peak at $\omega_1$ equal to the Nyquist frequency. Here, the highest maxima outside frequency 0 are at $\omega_1 =$ 1/0.997 day\textsuperscript{-1} and $\omega_2 =$ 1/31 day\textsuperscript{-1}, and  We find in 13 out of 17 Bayes factor false detections occur at $\omega \pm \omega_1$ or $\omega_1$.
 Because it is clear that there is a planet, the Bayes factor detection of aliases might be very significant, and the threshold has to be elevated very significantly to make these detections non significant, causing to loose viable detections of low amplitude signals. 
 On the contrary, false detections by the FIP are aliases of true periods in 2 cases out of 7, because it is able to determine when there is a  degeneracy between periods. As expected, this problem worsens as the signal-to-noise ratio decreases (see Fig.~\ref{fig:mistakes_roc_red} (a) and (b) near 1 false detection).

  As mentioned in section~\ref{sec:introduction}, our original concern was that existing methods de-couple the search for the period of the planet and its significance.  
  Our results confirm that, indeed, the scale of the significance metric has to be tied to the accuracy desired on the parameters, here on the period. In~\cite{hara2021}, we show that this is also the case in practice, through the analysis of the HARPS RV dataset of the star HD 10180~\citep{lovis2011}.






\subsection{Mixture models and label-switching}

In this work, we consider likelihood functions of the form $p(y \mid (\theta_i)_{i=1..n}, \eta)$ where $\theta_i$s are all in the same space. This form applies in particular mixture models, the likelihood is assumed to be a weighted sum of distributions
\begin{align}
  p(y \mid (\theta_i)_{i=1..n}, \eta) = \sum\limits_{i=1}^n \theta_i^{1}   p_0(y \mid \theta_i^{2:}, \eta)
\end{align}
where $\theta_i $ is a vector written $ (\theta_i^1, \theta_i^{2:})$, and $p_0$ is a likelihood function of a given form. Typically $p_0$ is Gaussian, $\theta_i$s are three dimensional vectors, $\theta_i^2$ and $\theta_i^3$ are the mean and variance. 
 In mixture models, the likelihood is invariant by permutation of the labels of the $\theta_i$. For any permutation $\nu$, $p(y \mid (\theta_i)_{i=1..n}, \eta) = p(y \mid (\theta_{\nu(i)})_{i=1..n}, \eta)$. If the prior does not distinguish between the $\theta_i$s, then the posterior is also invariant by label-switching, which might cause difficulties in its interpretation \citep[e. g.][]{redner1984, titterington1985, celeux, stephens2002, celeux2018}, in particular to infer the number of components in the mixture~\citep{stephens2002,roodaki2014}. 

For that latter problem, if we want not only to estimate the number of components but a certain accuracy is desired on their parameters, then the theory presented above is readily applicable.  One simply has to define the detection of mixture components as in Definition~\ref{def:detec}, and following the analysis of Section~\ref{sec:procedure} is sufficient. As in~\cite{stephens2002}, our solution to label-switching degeneracies relies on a decision-theoretic framework: label switching does not pose a major problem conceptually as long as it is clear what is to be inferred from the posterior. 
 
As a remark, in the case of exoplanets it is  possible to break the label-switching degeneracy by imposing an order on the periods, which also has the advantage of speeding up computations~\citep{nelson2020}. However, ordering might be insufficient to interpret certain posterior distributions in mixture models \citep[e.g.][]{stephens2002}.


\subsection{Bayesian false discovery rate }
\label{sec:fdr}

In Section~\ref{sec:framework} and ~\ref{sec:constraint}, we respectively pose the problem in terms of maximisation of utility function and miminimum missed detections under constraint on the expected number of false one.~\cite{muller2004} uses the same definition of functions to optimize  in a different context: gene expression microarrays. They consider $n$ different genes which might have an impact or not not the result of some experiment, and want to decide for gene $i$ if it has a significant impact or not. In their parametrization, a variable $z_i$ plays the role of the ground truth, $z_i=0$ means the gene has no impact and $z_i=1$, that it has. The optimal decision rule consists in selecting the the $n^\star$ genes with posterior probability of $z_i=1$ greater than a certain threshold. They use this result to determine the appropriate sample size. Other works also consider discrete hypotheses, such as ~\citep{guindani2009}.
~\cite{barbieriberger2004} consider the problem where the data $y$ are described by a linear model $y = Ax + \epsilon$, where each of the components of vector of size $n$, $x$, might be zero. They aim at finding the way to select the non zero components to have optimal predictive properties, and find that they must select the components with the probability to be non-zero greater than 1/2.

 For our problem, we have to deal with a continuum of possibilities. If we choose the $\Theta_i$ as balls of fixed radius and centre $\theta$, our hypotheses are indexed by the continuous variable $\theta$. 
 We can however adapt our framework to deal with the cases of~\cite{muller2004} by selecting $T$ (see Section~\ref{sec:mathprob}) as a discrete set of indices $1..n$, and the case of~\cite{barbieriberger2004} as the set of $(i,\theta_i)$, where $\theta_i$ is a linear coefficient, $i=1,..,n$. 
 We must further impose that the prior forbids to choose twice the same index, similarly to our separability condition (Definition~\ref{def:separability}).
~\cite{muller2004} show that the optimal detection criterion consists in selecting the $n^\star$ genes with a probability $P(z_i=1 \mid y) > c/(c+1)$ where $c$ controls the relative cost of false positives and false negatives (in our formalism, $\gamma =1/c$) which is exactly the same as Eq.~\eqref{eq:criterionFIP_fixed}.

In~\cite{muller2006}, it is shown that the decision rule of~\cite{muller2004} based on a decision-theoretic approach has close links with Bayesian approaches to false discovery rate (FDR)~\citep{efron2001} (see also~\cite{storey2003, scott2003, bogdan2008, stephens2016}). The notion of FDR, introduced by~\cite{benjamini1995} aims at controlling the proportion of false positives among the signals whose detection is claimed. As noted 
in~\cite{hara2021}, the FIP provides guarantees on the number of false detections. Indeed, if the prior reflect the true underlying population of objects, and the likelihood is an accurate representation of the data,  among $n$ statistically independent claims made with FIP $= \alpha$, the number of false detections follows a binomial distributions of parameters $\alpha$ and $n$. However, this property is not guaranteed if the model assumptions are incorrect. In the next section, we suggest ways to diagnose incorrect assumptions on the prior and likelihood. 

\section{Model criticism}
The optimal decision is chosen within a set of alternatives, which can be nested in a hierarchical Bayesian models. For instance, in the case of exoplanets there are many different ways to parametrize the noise due to the star~\citep{aigrain2012, foremanmackey2017, delisle2019b}, whose respective Bayesian evidences can be compared \citep[e.g.][]{ahrer2021}. Instead of choosing a particular noise model (a particular likelihood), the FIP can simply be marginalised over the different noise models in a model averaging perspective. However, as noted in~\cite{rubin1984}, including models in increasingly deeper hierarchical structures has a computational limit. In~\cite{hara2021}, we showed numerically that the FIP exhibits a certain robustness to model misspecification, and here we aim to go further.

In the terms of~\cite{box1980}, we tackled the estimation problem (finding a decision in a set of alternatives), and now turn to the criticism problem: is the set of alternatives realistic? In the context of exoplanets, this latter problem is very seldom examined \citep[see tests on the residuals suggested in][]{baluev2013, hara2019}. 

Here, our priority is to gain confidence  that our predictions are calibrated: statistically independent detections made with FIP $= \alpha$ should behave as independent Bernoulli variables with probability of success, when there is indeed a planet, of $1-\alpha$. 
Calibration is deemed desirable in different contexts: statistical inference \citep[e.g.][]{box1980, dawid1982, rubin1984, gneiting2007, draper2010}, predictions in a game-theoric setting \citep[e.g.][]{fostervohra1998, vovkshafer2005} and machine learning  \citep[e. g.][]{song2019}. Even if the events which are assigned a probability may not be repeatable, the model (priors and likelihood) can be viewed as an ``expert'' which will take decisions repeatedly, and we want the expert to make calibrated predictions. If it is not the case, we want to understand which assumptions of the model are faulty.   
\label{sec:mopen}
\subsection{General method}

Calibration is often considered in a supervised manner: the truth can be known. For instance, when assessing the calibration of a method predicting the weather, it can be known for sure whether in fact it was sunny or rainy \citep[e.g.][]{fostervohra1998, vovkshafer2005}.  
However, for exoplanets it cannot be ensured with certainty, and rarely with extremely high probability that a decision is correct. 

To circumvent this issue, we adopt the same Bayesian cross validation formalism as~\cite{vehtari2002} (see also~\cite{draper2010}). The data $y$ on which the inference is based, whenever possible, is divided in two. The first part (training set) is used to compute a posterior for the predictive distribution on the second part (test set), which is then compared to the actual data.  For instance, if $y$ is a time series $y(t_i)_{i=1..N}$, we can separate it in a training sets with $y_\mathrm{train } = \{y(t_i), t_i<t_0\}$ and a test set  $y_\mathrm{test } = \{y(t_i), t_i>t_0\}$ for some threshold $t_0$, and compare $y_\mathrm{test}$ to the distribution of what should be observed at times $t_\mathrm{test } = t_i>t_0$, knowing $y_\mathrm{train}$, which we denote by $p(y^\star \mid y_\mathrm{train})$.

Let us suppose we are able to generate realisations of $y^\star$ following  $p(y^\star \mid y_\mathrm{test})$. We compute several statistics $S_i(y^\star), i=1..p$, where each $S_i$ retrieves a real number, and compare their distributions to $S_i(y_\mathrm{test})$. In~\cite{vehtari2002}, such statistics are defined as utilities expressing the discrepancy between the predicted data and the actual one. To ensure the predictions are calibrated, we envision two possibilities. As in~\cite{gneiting2007}, we can leverage that, denoting by $F_X$ the cumulative distribution function of random variable $X$, $F_{S_i}^{-1}(S_i(y_\mathrm{test}))$ should be distributed uniformly on [0,1]. Second, we can define $S_i(y^\star)$ as Bernoulli variables with probability of success, $S_i(y^\star)=1$, equal to $p_i$. The total number of successes on the $n$ is on average $\sum_{i=1,..,n} p_i$. If the $S_i$s are chosen to be independent, the number of successes follows a Binomial distribution, otherwise the distribution can be estimated with the empirical joint distribution of $(S_i(y^\star))_{i=1,..,n}$, obtained from the different realisations of $y^\star$ following $p(y^\star \mid y_\mathrm{test})$. As a remark, the number of relevant tests is limited by the Data processing inequality~\citep{coverthomas2006}, $I((S_i(y^\star))_{i=1,..,n};y_\mathrm{train}) \leqslant I(y^\star;y_\mathrm{train})$, where $I$ is the mutual information.  

The prior predictive test proposed by~\cite{box1980} is a particular case of our cross-validation test when the training set is the empty set, and the statistic considered is the Bayesian evidence. The posterior predictive test of~\cite{rubin1984, gelman1996}, conveys the idea that, if the exact same experience was to be repeated, we would want the actual observations not to be too improbable an outcome, and chooses $y_\mathrm{test} = y_\mathrm{train} =y$.  

To assess the calibration of our detections, we want to express that if a component has been detected in the training set, it should still be present in the test set. A natural way to define the statistics would be to determine on $y_\mathrm{train}$ the sets $\Theta_i \subset T, i=1,..,n$  maximizing the utility.  Suppose we draw $y^\star$ from the predictive distribution. Conditioning on $y^\star$, we obtain posterior probabilities to have planets in each $\Theta_i $, $p_i(y^\star)$, which can be our statistics $S_i(y^\star) = p_i(y^\star)$.  Now as $y^\star$ follows the  predictive distribution $p(y^\star \mid y_\mathrm{train })$, we obtain samples of the posterior distribution of $p_i(y^\star)$. Unfortunately, such calculations would in general require to evaluate $p_i(y^\star)$ for many realizations of $y^\star$. At least for exoplanets, this is unrealistic. Furthermore, these statistics alone do not allow to diagnose which assumptions of the model should be changed. How the data should be split, and which statistics to choose are difficult questions, and should be decided on a case by case basis.

\subsection{Example}
In the case of radial velocities, we suggest to define the statistics $S_i$ based on the Fourier transform (or periodogram) of the predicted data.
To illustrate this, let us suppose that conditioned on a time series $y_\mathrm{train}$, at times $t_\mathrm{train}$, the predictive distribution of data acquired at times  $t_\mathrm{test}$,  $p(y^\star \mid y_\mathrm{train})$, is such that we can draw samples from it as  
\begin{align}
    y^\star(t_\mathrm{test}) = A \cos \omega t_\mathrm{test} + B \sin \omega t_\mathrm{test}  + \epsilon \label{eq:ystar}
\end{align}
where A, B, $\omega$ are drawn from their posterior distribution knowing $y_\mathrm{train}$, which are assumed to be Gaussian distributions of mean and standard deviation $\mu_A, \sigma_A$, $\mu_B, \sigma_B$ and $\mu_\omega, \sigma_\omega$ respectively. We further assume that $\epsilon$ is a homoskedastic Gaussian noise of zero mean and standard deviation $\sigma$. 

We can define the test statistic as a Bernoulli variable, $T_1(y^\star) =1$ if the modulus of the Fourier transform of $y^\star$ at frequency $\mu_\omega$ is within certain bounds, and 0 otherwise. This expresses the idea that the putative planetary signal is still present. We can also define a series of test $S_i(y^\star) = 1$ if  the amplitude of the Fourier transform of $y^\star$ at frequency $\omega_i, i=1,..,p$ is within certain bounds. This conveys the idea that the noise spectrum should be flat.

To fix ideas, we take $t_\mathrm{test}$ as a hundred equispaced measurements, assuming that they are taken each day. We generate 100,000 realisations of $y^\star$ according to Eq.~\eqref{eq:ystar}, and compute the centred 90\% credible intervals of each Fourier coefficient of $y^\star$. In Fig.~\ref{fig:test_fourier} (a), points show the absolute values of the Fourier coefficient for one of the 100,000 realisations. Four out of a hundred coefficients are outside of the 90\% intervals. Because of random fluctuations, we expect a stochastic variation of the number of coefficients outside bounds.  In Fig.~\ref{fig:test_fourier} (b), we show the distribution of the number of the number of coefficients outside bound, and four is not unusual. Let us now compute the absolute values of the Fourier coefficients of a realisation of $y^\star$ such that instead of generating white noise of variance 1 $m^2/s^2$ on top of a planetary signal, we generate a correlated Gaussian stationary noise  with an exponential autocorrelation $\mathrm{e}^{-\Delta t/\tau}$ with a timescale $\tau =3$ days, such that the diagonal has also a variance of 1 $m^2/s^2$. In that case, we find as expected higher and lower amplitudes at low and high frequency respectively, and twelve tests are outside of bounds.  
\begin{figure}
    \centering
    \includegraphics[width=\linewidth]{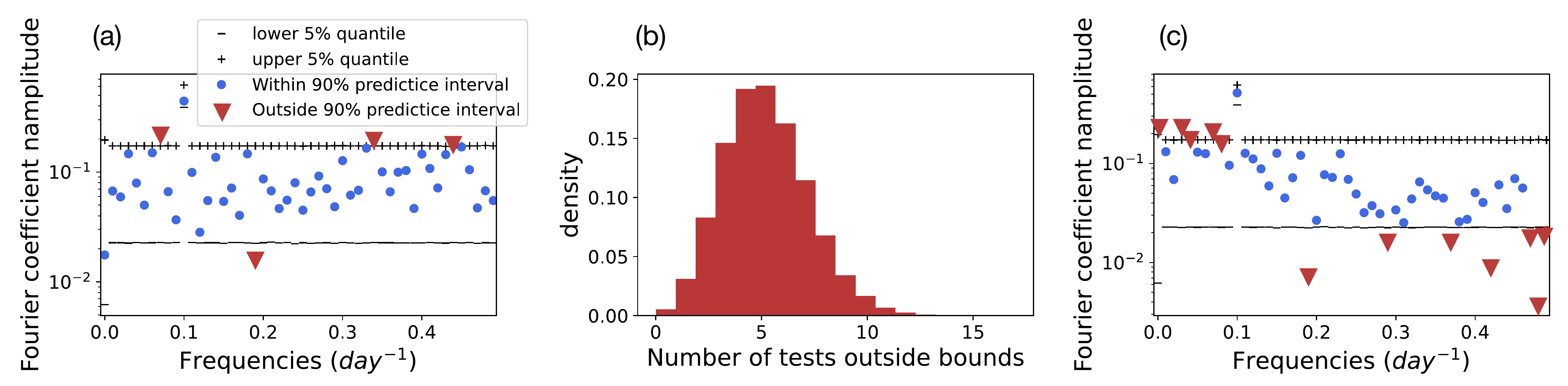}
    \caption{In (a) and (c) we show with markers - and + the bounds of the intervals where the modulus of Fourier coefficients should be as a function of frequency. Round and triangle markers represent the modulus of Fourier coefficients of a certain realisation of predicted data. Points inside and outside the 90\% interval bounds are shown with round and triangle markers, respectively. Figures (a) and (c) correspond to a realisation of the predictive distribution, and a Gaussian stationary noise  with an exponential autocorrelation $\mathrm{e}^{-\Delta t/\tau}$ with a timescale $\tau =3$ days, respectively. Figure (b) represents the histogram of the numbers of Fourier moduli outside bounds obtained with the 100,000 realisations of the predictive distribution. }
    \label{fig:test_fourier}
\end{figure}

To refine our tests, we can take as statistics the amplitudes of the putative planets fitted with a noise model, and the amplitudes of the Fourier transforms of the noise after fitting the planetary signals. The predicted phase of the putative planets can be examined, which offers an alternative to the phase consistency tests of~\cite{hara2022}.


\section{Conclusion}
\label{sec:conclusion}
Our initial aim was to find which parametric components are present in a dataset. We assume that there are $n=0$ to $n_\mathrm{max}$ components, with parameters $(\theta_i)_{i=1...n} \in T^n$ for some space $T$, and there are nuisance signals parametrized by $\eta$.
For data $y$, we assume that a likelihood $p(y\mid n,(\theta_i)_{i=1...n}, \eta)$ has been defined as well as proper priors on $n,(\theta_i)_{i=1...n}, \eta$. We considered detection claims of the form: ``there are $n$ components, one with parameters in $\Theta_1$, ..., one with parameters in $\Theta_n$'' where the $\Theta_i$s are regions of the parameter space belonging to a predefined family of alternatives. For each $i$ such that there is no such component truly with parameters in $\Theta_i$, we count a false detection. Conversely, underestimating the number of components or not finding true ones corresponds to missed detections. This formulation requires the data analyst to choose the candidate regions $\Theta_i$. If $T$ is a metric space, we recommend to choose a set of the closed balls of fixed radius $R$, which fixes the resolution desired on the component parameters. If possible, we recommend to choose $R$ such that two components cannot be closer than $2R$. The optimal detection criterion can be computed as a by-product of the calculations necessary to evaluate Bayesian evidences of models with different numbers of components \footnote{See also the Python notebook example in supplementary material~\citep{hara2023_supp_code}, also available at \url{https://github.com/nathanchara/FIP}}. 

\subsection{General case}

We found that, essentially, provided the components cannot have parameters too close to each other, which we formalise as the separability condition~\ref{def:separability}, the optimal detection procedure consists in selecting the regions $\Theta_i$ with a sufficient posterior probability to ``host'' the parameters of a component (see Eq.~\eqref{eq:criterionFIP_fixed}), and our approach can be seen as an extension of Bayesian false discovery rate to continuous indices (see section~\ref{sec:fdr}). If the separability condition is not met, the problem can be solved with a more general formulation (see section~\ref{sec:utility_some}).
Our significance metric is a posterior probability, which is particularly meaningful if it is calibrated: detections made with probability $\alpha$ are correct, on average, in a fraction $\alpha$ of cases. We leverage this property to test whether the prior and likelihood chosen are realistic. 


Both the optimal criterion discussed in sections~\ref{sec:fip}-\ref{sec:discussion} and the calibration discussed in Section~\ref{sec:mopen}, can be seen through an uncertainty principle. In the first case, the smaller the $\Theta_i$s are chosen, the more resolution we have on the component parameters, but the less probability mass will be contained in the $\Theta_i$s. Similarly, we might ensure calibration using the data of a whole survey (for example, several hundreds of stars observed with the same instrument), or we might want to ensure calibration on a subset of events: for detections around certain type of stars, detections of small planets etc. This restricts the number of tests that can be run, thus the statistical power of calibration tests.

From a theoretical perspective, the next step will be to determine an analytical approximation of the minimum  expected number of missed detections under a constraint on expected number of false detections, as a function of the prior distribution of parameters and the likelihood. This way, it will be possible to predict the optimal capabilities of a survey operating at a certain precision, assuming a certain parametric form for the priors and likelihood. 

\subsection{Exoplanets and astronomy}

For the detection of exoplanets with radial velocities, instead of using the Bayes factor to determine the number of planets, it is more advantageous to proceed as follows. We define $\Delta \omega = 2\pi/T_{obs}$ where $T_{obs}$ is the observational time-span, and a grid of tightly spaced frequencies $\omega_k$ which defines a collection of intervals $I_k = [\omega_k -\Delta \omega/2, \omega_k -\Delta \omega/2 ]$. One has to evaluate the posterior probability to have a planet with frequency in $I_k$ (TIP) or the FIP = 1-TIP, which is computationally challenging~\citep{nelson2020}. The next step is to select the  maximum number $n^\star$ of disjoint intervals with a FIP satisfying Eq.~\eqref{eq:criterionFIP_fixed}. We provide a Python example\footnote{\label{note1}Available in supplementary material~\citep{hara2023_supp_code}, and  \url{https://github.com/nathanchara/FIP}.}.

For other exoplanet observation techniques, if a correct detection is defined as one of the parameters being in a certain region, the optimal procedure is the same. Additional work is needed to make our framework computationally tractable for transits, imaging and microlensing, where datasets are typically much more voluminous. It will also be interesting to test the method in other contexts, such as the detection of gamma ray emissions \citep[e.g.][]{geringersameth2015}.

To detect exoplanets, different models are compared to each other (different number of planets, different noise models). The present work gives a criterion with optimality properties once a set of  alternative models have been considered. So far, whether the ensemble of alternatives considered is plausible, or if all the alternatives are poor, is almost never addressed. In Section~\ref{sec:mopen}, we propose a method to tackle this question, based on the predictive distribution. Such tests are to be refined in future work.




\section*{Acknowledgements}
N. C. Hara thanks Gwenaël Boué for his insightful remarks, Roberto Trotta for his suggestions, Michaël Crétignier for providing the HD21693 HARPS data, Julien Bect for pointing out resemblances with the label-swtching problem, and A. Leleu for his comments on dynamical stability.
N. C. Hara and J.-B. Delisle acknowledge the financial support of the National Centre for Competence in Research PlanetS of the Swiss National Science Foundation (SNSF). T. de Poyferr\'e aknowledges support by the NSF under Grant No.DMS-1928930 while participating in a program hosted by MSRI in Berkeley, Ca, during Spring 2021.

\begin{supplement}
\stitle{Mathematical proofs~\citep{hara2023_supp}}
\sdescription{Contains the proofs of the mathematical results presented in the main text. It also contains a reproduction of the definitions of the main text necessary to the proofs. }
\end{supplement}

\begin{supplement}
\stitle{Python notebook example of FIP calculation~\citep{hara2023_supp_code}}
\sdescription{Notebook that provides an example of FIP calculation based on samples of a nested sampling algorithm. These samples were generated for the analysis of the HD93385 HARPS radial velocity data made in~\cite{unger2021}. This supplementary file is made available in the form of a single zip file which contains the Python code and the samples to be processed.}
\end{supplement}

\appendix

\section{Expression of the utility function}
\label{app:b}
In this section, we show how to obtain the expression of the utility function~\eqref{eq:ujplagamma_othercriterion}, and the constrained problem, Eq.~\eqref{eq:constraints_othercriterion}. We also obtain the expression of the utility function when using the alternative definition of missed detections is:
 \begin{definition}[Missed detections: other definition] 
If $n$ components are claimed, but in fact there are $n'>n$ components truly present in the data, we count $n'-n$ missed detections.
 \label{def:misseddec1}
 \end{definition}

\subsection{Expression of the utility function, and expected number of false and missed detection for the definition of missed detections~\ref{def:misseddec2}}
\label{app:b1}

The expected utility is the term $ E_{\theta,\eta}\left[ U\left\{a(\Theta_1, ..., \Theta_n), (\theta,\eta)\right\} \right]$ in Eq.~\eqref{eq:fullprob}. For the definition of missed detections given in Definition~\ref{def:misseddec2}, The expected utility of choosing $\Theta_1,...,\Theta_n$ is given in Eq. (10) of the main text, reproduced below. 
\begin{align}
 E_{\theta,\eta}\left[ U\left\{a(\Theta_1, ..., \Theta_n), (\theta,\eta)\right\} \right]   =  -n + (1 + \gamma) \sum\limits_{j=1}^n j I_{A_j}(\Theta_1, ..., \Theta_n) -  \gamma  \sum\limits_{k=1}^{n_{max}} k p(k\mid y), \label{eq:ujplagamma_othercriterion_app_supmat}
\end{align}
where $p(k\mid y)$ is defined in Eq.~\eqref{eq:pk} and $I_{A_j}(\Theta_1, ..., \Theta_n)$ is the probability that exactly $j$ detections are correct when choosing $\Theta_1,...,\Theta_n$. 

The constrained optimization problem given in Eq.~\eqref{eq:constraints_othercriterion} is 
\begin{align}
\argmax\limits_{(n,\Theta_1,..\Theta_n)}  \sum\limits_{j=1}^n j I_{A_j}(\Theta_1, ..., \Theta_n)
 \text{  subject to }  n- \sum\limits_{j=1}^n j I_{A_j(\Theta_1, ..., \Theta_n)} \leqslant x.
 \label{eq:constraints_othercriterion_supmat}
\end{align}


To prove Eq.~\eqref{eq:ujplagamma_othercriterion_app_supmat} and Eq.~\eqref{eq:constraints_othercriterion_supmat}, let us first compute the utility function for the claim ``there is no component'', denoted by $a_0$. 
\begin{align}
E_\eta\{U(a_0,\eta)\}  & = 0\times p(0\mid y) - \beta \sum\limits_{k=1}^{n_{max}} k p(k\mid y) .
\label{eq:u0pla}
\end{align}
If in fact there are $k$ components, we ``pay'' $k\beta$, hence the $-\beta \sum_{k=1}^n k p(k\mid y)$ term, where $p(k\mid y)$ the posterior probability of having $k$ components.

When claiming the detection of $n>0$ components with parameters in $\Theta_i$, $i=1...n$, the $\Theta_i$s can be considered as $n$ ``boxes''. We want to evaluate the utility of this claim if the true components have parameters $\theta_1,...,\theta_k$, where $k$ might be different from $n$. We consider ways to put the $\theta_i$ in the boxes such that each $\theta_i$ can go into only one ``box''. 
This simply expresses the fact that we can't claim twice the detection of a given component.  Furthermore, if a $\theta_i$ belongs to several boxes, there might be several ways to inject the $\theta_i$s into the boxes, yielding different number of correct detections. Suppose we claim that there are two components, one in  $\Theta_1$, one in $\Theta_2$, with a non empty intersection. Suppose that the true parameters are such that $\theta_1$ belongs to $\Theta_1 \cap \Theta_2$ and $\theta_2$ belongs to $\Theta_2$. Depending on whether we associate $\theta_1$ to $\Theta_1$ or $\Theta_2$, we have two or only one correct detections. In that case, we choose the injection which leads to as many correct detections as possible. 

We denote by $m$ the maximum number of different $\theta_i$s that we can put in a $\Theta_i$. 
We denote by $A_m^k$ the region of parameter space with $k$ components 
such that there is exactly one component in each of the $\Theta_i$, $i=1..m$, $m\leqslant n$.

 If $k$ planets are truly present in the data, $n$ detections are claimed but only $i$ are correct, it means that the true detections of  $\min (k,n)-i$ are missed. We can penalize this situation by adding a term $-\beta(\min (k,n)-i)$ whenever it happens. The expression of the utility function
\begin{align}
& E_{\theta,\eta}\left[ U\left\{a, (\theta,\eta)\right\} \right]   =  -n \alpha p(0 \mid y )  \label{eq:0pla_n2}\\ &+\left[ -(n-1)\alpha I_{A_1^1} - (n\alpha+\beta)   \left(1 -  I_{A_1^1}\right)\right] p(1\mid y )  \label{eq:1pla_n2} \\ 
&+   \left[ -(n-2)\alpha I_{A_2^2} -((n-1)\alpha+\beta) I_{A_1^2} - (n\alpha +2\beta)  \left(1 -  I_{A_1^2} -   I_{A_2^2}\right)\right] p(2\mid y )      \label{eq:2pla_n2} \\
& \vdots  \\
&+  \left[ \sum\limits_{i=1}^k (-(n-i)\alpha -(k-i)\beta)I_{A_i^k} -  (n\alpha+k\beta) \left(1 - \sum\limits_{i=1}^k I_{A_i^k}\right)  \right] p(k\mid y ) \label{eq:npla_k} \\
& \vdots  \\
&+  \left[ \sum\limits_{i=1}^n -(n-i)(\alpha +\beta)I_{A_i^n} -  n(\alpha+\beta) \left(1 - \sum\limits_{i=1}^n I_{A_i^n}\right)  \right] p(n\mid y ) \label{eq:npla_n2} \\
&+  \left[ \sum\limits_{i=1}^n -(n-i)(\alpha +\beta) I_{A_i^{n+1}} - n(\alpha +\beta) \left(1 - \sum\limits_{i=1}^n I_{A_i^{n+1}}\right) -  \beta \right] p(n+1\mid y ) \label{eq:npla_np12} \\
& \vdots  \\
&+  \left[ \sum\limits_{i=1}^n -(n-i)(\alpha +\beta) I_{A_i^n} -  n(\alpha +\beta) \left(1 - \sum\limits_{i=1}^n I_{A_i^{n_{max}}}\right) - (n_{max}-n) \beta \right] p(n_{max}\mid y ) \label{eq:npla_nmax2}
\end{align}
Re-arranging the terms, we have 
\begin{align}
 E_{\theta,\eta}\left[ U\left\{a, (\theta,\eta)\right\} \right]  & =  -n\alpha + (\alpha + \beta) \sum\limits_{i=1}^n i I_{A_i} -  \beta  \sum\limits_{k=1}^{n_{max}} k p(k\mid y)  \\& =  - (\alpha E[\mathrm{FD}] +  \beta  E[\mathrm{MD}])
\label{eq:unpla_2}
\end{align}
where $ E[\mathrm{FD}] $ and $ E[\mathrm{MD}]$ are the expected numbers of false detections and missed detections when claiming the detection of components with parameters in $\Theta_1,...,\Theta_n$, respectively,
\begin{align}
     E[\mathrm{FD}]& = n - \sum\limits_{i=1}^n i I_{A_i} \\
     E[\mathrm{MD}] &= \bar{n} - \sum\limits_{i=1}^n i I_{A_i},
\end{align}
where $\bar{n} := \sum_{k=1}^{n_{max}} k p(k\mid y)$ does not depend on the number of component claimed $n$.
Assuming that $\alpha \neq 0$ (or equivalently $\alpha > 0$, since $\alpha $ is non negative), we can divide Eq.~\eqref{eq:unpla_2} by $\alpha$. Denoting by $\gamma = \beta/\alpha$, without loss of generality we can maximize
\begin{align}
 E_{\theta,\eta}\left[ U\left\{a, (\theta,\eta)\right\} \right]   =  -n + (1 + \gamma) \sum\limits_{i=1}^n i I_{A_i} -  \gamma  \bar{n}. 
\label{eq:ujplagamma_othercriterion_app}
\end{align}
Where $\bar{n}$ does not depend on the number of planets. 

\subsection{Expression of the utility function, and expected number of false and missed detection for the definition of missed detections~\ref{def:misseddec1}}

\label{app:b2}

In that case the expected utility is 
\begin{align}
& E_{\theta,\eta}\left[ U\left\{a, (\theta,\eta)\right\} \right]   =  -n \alpha p(0 \mid y )  \label{eq:0pla_n}\\ &+\left[ -(n-1)\alpha I_{A_1^1} - n\alpha   \left(1 -  I_{A_1^1}\right)\right] p(1\mid y )  \label{eq:1pla_n} \\ 
&+   \left[ -(n-2)\alpha I_{A_2^2} -(n-1)\alpha I_{A_1^2} - n\alpha   \left(1 -  I_{A_1^2} -   I_{A_2^2}\right)\right] p(2\mid y )      \label{eq:2pla_n} \\
& \vdots  \\
&+  \left[ \sum\limits_{i=1}^n -(n-i)\alpha I_{A_i^n} -  n\alpha \left(1 - \sum\limits_{i=1}^n I_{A_i^n}\right)  \right] p(n\mid y ) \label{eq:npla_n} \\
&+  \left[ \sum\limits_{i=1}^n -(n-i)\alpha I_{A_i^{n+1}} -  n\alpha \left(1 - \sum\limits_{i=1}^n I_{A_i^{n+1}}\right) -  \beta \right] p(n+1\mid y ) \label{eq:npla_np1} \\
& \vdots  \\
&+  \left[ \sum\limits_{i=1}^n -(n-i)\alpha I_{A_i^n} -  n\alpha \left(1 - \sum\limits_{i=1}^n I_{A_i^{n_{max}}}\right) - (n_{max}-n) \beta \right] p(n_{max}\mid y ) \label{eq:npla_nmax}
\end{align}
Re-arranging the terms, we have
\begin{align}
 E_{\theta,\eta}\left[ U\left\{a, (\theta,\eta)\right\} \right]  & =  -n\alpha  + \alpha \sum\limits_{i=1}^n i I_{A_i} -  \beta  \sum\limits_{k=n+1}^{n_{max}} (k-n) p(k\mid y) \\
& =  - (\alpha E[\mathrm{FD}] +  \beta  E[\mathrm{MD}])
\label{eq:unpla}
\end{align}
where $ E[\mathrm{FD}] $ and $ E[\mathrm{MD}]$ are the expected numbers of false detections and missed detections, respectively. 
\begin{align}
     E[\mathrm{FD}]& = n - \sum\limits_{i=1}^n i I_{A_i} \\
     E[\mathrm{MD}] &= \sum\limits_{k=n+1}^{n_{max}} (k-n) p(k\mid y)
\end{align}

Assuming that $\alpha \neq 0$ (or equivalently $\alpha > 0$, since $\alpha $ is non negative), we can divide Eq.~\eqref{eq:unpla} by $\alpha$. Denoting by $\gamma = \beta/\alpha$, without loss of generality we can maximize
\begin{align}
 E_{\theta,\eta}\left[ U\left\{a(\Theta_1, ..., \Theta_n), (\theta,\eta)\right\} \right]  =  -n  + \sum\limits_{j=1}^n j I_{A_j} -  \gamma  \sum\limits_{k=n+1}^{n_{max}} (k-n) p(k\mid y) .
\label{eq:ujplagamma_app}
\end{align}
The expected utility is 
\begin{align}
E_{\theta,\eta}\left[ U\left\{a(\Theta_1, ..., \Theta_n), (\theta,\eta)\right\} \right]   =  -n  + \sum\limits_{j=1}^n j I_{A_j}(\Theta_1, ..., \Theta_n) -  \gamma  \sum\limits_{k=n+1}^{n_{\mathrm{max}}} (k-n) p(k\mid y),
\label{eq:ujplagamma}
\end{align}

For the definition of missed detections adopted in Definition.~\ref{def:misseddec1}, the constrained problem writes 
\begin{align}
\argmin\limits_{(n,\Theta_1,..\Theta_n)} \sum\limits_{k=n+1}^{n_{\mathrm{max}}} (k-n) p(k\mid y)
 \text{  subject to }  n- \sum\limits_{j=1}^n j I_{A_j} \leqslant x .
 \label{eq:constraints}
\end{align}
Among the solutions to this problem, we further select the ones that minimise the term  $n- \sum_{j=1}^n j I_{A_j}$. The rationale is that, for a given value of the objective function, there is no reason to select a solution with a higher expected number of false detection than necessary.

\section{Proof of Lemma 3.1}

\label{app:c}

In this section we prove Lemma 3.1 of the main text, reproduced below.
\label{app:lemma1}
\begin{lemma}
 Let us consider $\Theta_1 \in T,...,\Theta_n \in T$    , $ \forall i_1, i_2 = 1...n, i_1 \neq i_2 $, $\Theta_{i_1} \cap \Theta_{i_2} = \emptyset$. Denoting by 
 $A_i$  the subset of $\Theta$ such that there are exactly $i$ components in $\Theta_1,...,\Theta_n$, then
\begin{align}
 \sum\limits_{j=1}^n j I_{A_j}  = \sum\limits_{i=1}^n \mathrm{TIP}_{\Theta_i} 
\end{align}
where $\mathrm{TIP}_{\Theta_i} $ is defined in Eq.~\eqref{eq:pa}.
\label{lem:separ_app}
\end{lemma}
 We begin with a technical remark. 
\begin{remark}
\label{rem:permutation}
Since all components are interchangeable in the model, the ordering chosen between them is of no consequence. The detection claims are invariant by relabeling of the parameters, between different factors of $T$, i.e. having exactly one component in $\Theta_1$,... exactly one in $\Theta_n$ is equivalent to having exactly one component in $\Theta_{\sigma(1)}$,... exactly one in $\Theta_{\sigma(n)}$
for any permutation $\sigma$ of the $n$ labels. 
\end{remark}
\begin{proof}
Let us denote by 
$ \Theta_1  \wedge \Theta_2 \wedge ...\wedge \Theta_j \backslash \Theta_{j+1}, ..., \Theta_n$ regions of $\Theta$ such that $j$ components are in $\Theta_1, \Theta_2,..., \Theta_j$  and no components are in $\Theta_{j+1}, ..., \Theta_n$. Then, since the $\Theta_i$ are disjoint , we can decompose $\mathrm{TIP}_{\Theta_i} $ as a sum of posterior mass over regions that have a component in $\Theta_i$ but not other component in one of the $\Theta_j$, $j\neq i$, regions that have a component in $\Theta_i$ and $\Theta_j$ but none in $\Theta_k$, $k\neq i$, $k\neq j$ and so on. 
\begin{align}
\mathrm{TIP}_{\Theta_i} =  \sum\limits_{j=0}^{n-1}  \sum\limits_{k_1,...,k_j \in \llbracket 1,n \rrbracket_j\backslash \{i\}} I_{\Theta_{i}  \wedge \Theta_{k_1} \wedge ...\wedge \Theta_{k_j} \backslash \Theta_{k_{j+1}}, ..., \Theta_{k_n} } 
\end{align}
where $\llbracket 1,n \rrbracket_j$ is a draw of $j$ indices without replacement in $\llbracket 1,n \rrbracket$. For $j=1..n$,
\begin{align}
I_{A_j} &=   \sum\limits_{k_1,...,k_j \in \llbracket 1,n \rrbracket_j } I_{\Theta_{k_1}  \wedge \Theta_{k_2} \wedge ...\wedge \Theta_{k_j} \backslash \Theta_{k_{j+1}}, ..., \Theta_{k_n} }  
\end{align}
Then
\begin{align}
\sum\limits_{i=1}^n \mathrm{TIP}_{\Theta_i} = \sum\limits_{i=1}^n\sum\limits_{j=0}^{n-1}  \sum\limits_{k_1,...,k_j \in \llbracket 1,n \rrbracket_j\backslash \{i\}} I_{\Theta_{i}  \wedge \Theta_{k_1} \wedge ...\wedge \Theta_{k_j} \backslash \Theta_{k_{j+1}}, ..., \Theta_{k_n} } 
\end{align}
In this sum, the term $I_{\Theta_{1}  \wedge \Theta_{2} \wedge ...\wedge \Theta_{n}}$ appears $n$ times, the terms $I_{\Theta_{1}  \wedge ...  \wedge \Theta_{i-1} \wedge \Theta_{i+1} \wedge ...\wedge \Theta_{n}} $ appear $n-1$ times, so we obtain the desired result.
\end{proof}

\section{Existence of the solution}
\label{app:existence}
\label{app:d}

The existence of the solution to~\eqref{eq:maxn} can be guaranteed in the following situation. Let us suppose that $T$ is a metric space, and let us denote by $B_a$ a ball (closed or open) of fixed radius $L$  in $T$ centered on $a \in T$.   
\begin{lemma}
    Assume $T$ is a finite dimensional Riemannian manifold
    (or more restrictively  
    a finite product of $\mathbb{S}^1$ (angles) and $\mathbb{R}$ (actions).)
	Let us suppose that   $\mathcal{A} = \{B_{a_1},..., B_{a_n}, k=1,..n_{max}, a_1,...,a_k \in T_0^k \}$ where $T_0$ is a compact subset of $T$. 
    Then the maximisation problem
   ~\eqref{eq:maxn} has a (not necessarily unique) solution.
\end{lemma}

\begin{proof}
Note first that if the posterior probability is regular enough (non-singular with respect to Lebesgue measure), which we assume, the problem consists in maximizing, for each 
$n$, a linear combination of integrals of this distribution over sets in $\mathcal{A}$. As the center $a_i$ move continuously, the integration sets $A_i$ move continuously (in the Hausdorff topology for instance), and integration over them is continuous. Thus we are maximizing a continuous functional. 

Let us first show that the problem has a solution for each fixed $n$. The set of candidates is a smooth manifold, and the dependence of the functional to maximize is through integrating a probability distribution over sets of fixed diameters. As the centers $a_i$ of the balls go to infinity, the probability distribution has to become very small, and so does its integrals over fixed-sized balls; thus the value of the function to maximize goes to $0$ has the parameters go to infinty. Since the functional is positive and non-zero, there is some $\epsilon>0$ such that the set on which the functional is bigger than $\epsilon$ is compact. The functional, being continuous, attains its maximum on this set and this is then necessarily a global maximum. Then the maximum for $0\leq n\leq n_{max}$ is a maximum of~\eqref{eq:maxn}. (Or if $n$ is not bounded, the values of the integrals of the posterior distribution as a function of $n$ have to decrease to $0$ uniformly since the whole probability distribution sums to 1, so by the same argument as above the maximum is attained for bounded $n$.)
\end{proof}

\section{Proof of Lemma 3.2 and Lemma 3.4}
\label{app:lemma2}

\label{app:e}

In this section we prove Lemma 3.2 and Lemma 3.4 of the main text, reproduced below.
\begin{lemma}[Lemma 3.2 in the main text]
 Let us suppose that~\eqref{eq:maxn} has a solution $\Theta_1^n \in T,...,\Theta_n^n \in T$ , with $\mathrm{TIP}_{\Theta_1^n} \geqslant ...\geqslant \mathrm{TIP}_{\Theta_n^n}$. 
Then the solution to $(P_{n+1})$ is either $(\Theta_1^n,..., \Theta_n^n, \Theta_{n+1}^\star )$ or such that $\forall i \in \llbracket 1,n+1\rrbracket$, $\exists j \in \llbracket 1,n\rrbracket$ such that $\Theta_{i}^{n+1} \cap  \Theta_{j}^{n} \neq \emptyset$.
\label{lem:find_app}
\end{lemma}
\begin{proof}
The proof relies on the simple property $(P1)$: if a function $f:E\rightarrow\mathbb{R}$ attains its maximum in a set $X$, then $\forall D \subset E$ such that $X \cap D =\emptyset $, the set of solution to $\argmax_{x \in E\backslash D} f(x)$ is $X$.

Let us consider $\Theta_{n+1} \in T$. The solution to $(P_{n+1})$ can be written as 
\begin{align}
    \argmax\limits_{\substack{\Theta_1 \in T\backslash \Theta_{n+1},... \Theta_{n} \in T\backslash \Theta_{n+1} \\ \forall i,j \in \llbracket 1,n \rrbracket , i\neq j, \Theta_i \cap \Theta_j = \emptyset } } \mathrm{TIP}_{\Theta_{n+1}} +\sum\limits_{i=1}^n    \mathrm{TIP}_{\Theta_i}.
\end{align}
Either $\forall i \in \llbracket 1,n\rrbracket$, $\Theta_{n+1}^{n+1} \cap \Theta_{i}^n = \emptyset$ then thanks to $(P1)$, for $E = T^n$ and $D = \{ x_1,...,x_n \in T^n, \forall i, x_i \notin \Theta_{n+1}^{n+1}  \}$
\begin{align}
    \argmax\limits_{\substack{\Theta_1 \in T\backslash \Theta_{n+1},... \Theta_{n} \in T\backslash \Theta_{n+1} \\ \forall i,j \in \llbracket 1,n \rrbracket , i\neq j, \Theta_i \cap \Theta_j = \emptyset }} \sum\limits_{i=1}^n    \mathrm{TIP}_{\Theta_i} = (\Theta_i^n)_{i=1..n}
\end{align}
As a consequence,
\begin{align}
    \argmax\limits_{\Theta_{n+1} \in T \backslash \cup_{i=1}^n \Theta_i^n } \argmax\limits_{\substack{\Theta_1 \in T\backslash \Theta_{n+1},... \Theta_{n} \in T\backslash \Theta_{n+1} \\ \forall i,j \in \llbracket 1,n \rrbracket , i\neq j, \Theta_i \cap \Theta_j = \emptyset } } \mathrm{TIP}_{\Theta_{n+1}} +\sum\limits_{i=1}^n    \mathrm{TIP}_{\Theta_i}  = 
    (\Theta_1^n,..., \Theta_n^n, \Theta_{n+1}^\star )
\end{align}
 up to a permutation of the indices (see remark B.2). 
If $ \exists i \in \llbracket 1,n+1 \rrbracket, \forall j \in \llbracket 1,n \rrbracket \in  \Theta_{i} \cap  \Theta_{j}^n = \emptyset $ then the same argument applies and the solution to $(P_{n+1})$ is $(\Theta_1^n,..., \Theta_n^n, \Theta_{n+1}^\star )$. 

Let us denote by $\neg P$ the negation of a proposition $P$. Since $ \neg (\exists i \in \llbracket 1,n+1 \rrbracket, \forall j \in \llbracket 1,n \rrbracket,  \Theta_{i} \cap  \Theta_{j}^n = \emptyset) $  = $ \forall  i \in \llbracket 1,n+1 \rrbracket, \exists j  \in \llbracket 1,n \rrbracket,  \Theta_{i} \cap  \Theta_{j}^n \neq \emptyset  $, and the union of the two cases account for all cases, we obtain the desired result. 
\end{proof}

We now consider the case where $T$ is a metric space and the $\Theta_i$s are balls of fixed radius. 

\begin{lemma}[Lemma 3.4 in the main text]
	Let us suppose that $\mathcal{T}$ is the set of 0 to $n_{\mathrm{max}}$ disjoint balls of radius $R$ and centres $(\theta_i)_{i=1...n}$. Denoting by  $\Theta^c = \cup_{i=1}^n B(\theta_i, 3 R) \backslash  \cup_{i=1}^n \Theta_i$, if $\mathrm{TIP}_{\Theta^c} < \mathrm{TIP}_{\Theta^\star_{n+1}} $, then the components are separable of order $n+1$.
	\label{lem:app_balls}
\end{lemma}

\begin{proof}
	$\Theta^c$ is the space described by any set of balls of radius $R$ with a non empty intersection with the $\Theta_i$, $i=1...n$. If $\mathrm{TIP}_{\Theta^c} < \mathrm{TIP}_{\Theta^\star_{n+1}} $ there cannot be $n+1$ disjoint regions $(\Theta_i^{n+1})_{i=1...n+1}$ with non zero intersection with $\Theta_i$, $i=1...n$ with $\sum_{i=1}^{n+1}\mathrm{TIP}_{\Theta_i^{n+1}} \geqslant \mathrm{TIP}_{\Theta_{n+1}^{\star}} +\sum_{i=1}^{n}\mathrm{TIP}_{\Theta_i^{n}} $. 
\end{proof}

\section{Proof of Theorem 4.1}
\label{app:link}
\label{app:f}

In this section, we prove Theorem 4.1 of the main text, re-stated below for clarity
\begin{theorem}[Theorem 4.1 in the main text]
Let us consider a dataset $y$ and suppose that it verifies component separability at all orders, then there exists an increasing function $\gamma(x) \geqslant 0$  such that ~\eqref{eq:ujplagamma_othercriterion_app_supmat} and~\eqref{eq:constraints_othercriterion_supmat}  have the same solution. 
\label{theorem_app}
\end{theorem}

We prove a slightly more general version, accounting for both definitions of missed detections, Definitions~\ref{def:misseddec2} and 
\begin{theorem}
Let us consider a dataset $y$ and suppose that it verifies component separability at all orders, $n=1...n_{\mathrm{max}}$ 
then there exists an increasing function $\gamma(x) > 0$ such that ~\eqref{eq:ujplagamma} and~\eqref{eq:constraints} have the same solution, and a function $\gamma'(x) > 0$ such that ~\eqref{eq:ujplagamma_othercriterion_app_supmat} and~\eqref{eq:constraints_othercriterion_supmat}  have the same solution. 
\label{theorem}
\end{theorem}

The proof  relies on the following idea. Defining $u_n$ and $v_n$ the expected number of false and missed detection when solving~\eqref{eq:maxn}, we prove that $u_n$ is increasing and $v_n$ is decreasing. Thus, solving the constrained problem consists in finding $n_0$, the maximum $n$ such that $u_n<x$. Denoting by $\Delta u_n = u_n - u_{n-1}$, $\delta v_n =  v_{n-1} - v_n $, and $w_n = \Delta u_{n}/ \delta v_{n}$, writing the utility $u_n + \gamma v_n$, we see that provided we can choose a $\gamma(x)$ in between $w_{n_0}$ and $w_{n_0+1}$, the maximum utility problem has the same solution as the constrained one. Now, we prove that under component separability, $\Delta u_n$  and $\delta v_n$ are increasing and decreasing, thus  $w_{n}$ is increasing and choosing an appropriate $\gamma(x)$ is possible. Furthermore, $w_{n_0}$ is an increasing function of $n_0$, which is an increasing function of $x$, so $\gamma(x)$ is an increasing function of $x$. 

Theorem~\ref{theorem_app}  assumes component separability (Definition~\ref{def:separability}). This assumption is stronger than necessary for some of the lemmas, and is not made by default. However, we assume throughout the section that the $\Theta_i$s are pairwise disjoint. Thanks to Lemma~\ref{lem:separ_app}, using  Eq.~\eqref{eq:fipdef}, we can write $n-\sum_{j=1}^n jI_{A_j} = \sum_{i=1}^n \mathrm{FIP}_{\Theta_i}$. We consider $\Theta_1^n,...,\Theta_n^n$, solving~\eqref{eq:maxn}, and define 
\begin{align}
    u_n &:= \sum\limits_{i=1}^n \mathrm{FIP}_{\Theta_i^n} \label{eq:un}\\
    v_n &:= \sum\limits_{k=n+1}^{n_{\mathrm{max}}} (k-n) p(k\mid y) \label{eq:vn} \; \; \; \; ; \; \;  \; \;   {v'}_n = \bar{n} - n +  \sum\limits_{i=1}^n \mathrm{FIP}_{\Theta_i^n}. 
\end{align}
where $\bar{n} = \sum_{k=1}^{n_{max}} k p(k\mid y)$.
The sequence $u_n$ is the expected number of false detections, $v_n$ and ${v'}_n$ are the expected number of missed detections for the missed detection  Definitions of~\ref{def:misseddec1} and ~\ref{def:misseddec2}, respectively. Note that $u_n$, $v_n$ and $v'_n$ depend on $y$, but we chose not to make that dependence explicit to simplify notations. 
\begin{lemma}
$\forall y $ in the sample space  the sequence $(u_n)_{n=1..n_{max}}$ is increasing, $(v_n)_{n=1..n_{max}}$ and $(v'_n)_{n=1..n_{max}}$ are  decreasing.
\end{lemma}
\begin{proof}
 Let us suppose that there exists $n$ such that $u_{n+1} < u_n$. That is 
\begin{align}
 n+1 - \sum\limits_{i=1}^{n+1}  \mathrm{TIP}_{\Theta_i^{n+1}}  <  n -   \sum\limits_{i=1}^{n}  \mathrm{TIP}_{\Theta_i^{n}}
\end{align}
This is equivalent to 
\begin{align}
 1 + \sum\limits_{i=1}^{n}  \mathrm{TIP}_{\Theta_i^{n} }<  \sum\limits_{i=1}^{n+1}  \mathrm{TIP}_{\Theta_i^{n+1}} 
\end{align}
Let us denote by $i_0$ an index such that $i_0 = \argmax_{i=1..n+1} \mathrm{TIP}_{\Theta_i^{n+1}}$. Then 
\begin{align}
1 \leqslant 1 + \sum\limits_{i=1}^{n}  \mathrm{TIP}_{\Theta_i^{n} }-\sum\limits_{i=1,i\neq i_0}^{n+1}  \mathrm{TIP}_{\Theta_i^{n+1}}  <   \mathrm{TIP}_{\Theta_{i_0}^{n+1}} 
\end{align}
Where the left inequality stems from the definition of the ${\Theta_i}^{n}$ and ${\Theta_i}^{n+1}$. Indeed, by definition 
$\sum_{i=1}^{n}  \mathrm{TIP}_{\Theta_i^{n}}$ is the maximum sum of probability mass on $n$ disjoint $\Theta_i$, 
\begin{align}
\sum\limits_{i=1}^{n}  \mathrm{TIP}_{\Theta_i^{n}} \geqslant \sum\limits_{i=1,i\neq i_0}^{n+1}  \mathrm{TIP}_{\Theta_i^{n+1}} .
\end{align}
We then have $1 <  \mathrm{TIP}_{\Theta_{i_0}^{n+1}} $, which is absurd. 

We have $v_{n+1} - v_n = - \sum_{k=n+1}^{n_{\mathrm{max}}}p(k \mid y) \leqslant 0$, and $v'_{n+1} - v'_n = u_{n+1} - u_n -1 = \sum_{i=1}^{n}  \mathrm{TIP}_{\Theta_i^{n}} - \sum_{i=1}^{n+1}  \mathrm{TIP}_{\Theta_i^{n+1}} \leqslant 0$ by definition of $\Theta_i^{n}$ and $\Theta_i^{n+1}$.
\end{proof}

Because of this result, we can now ensure that the solution to the constrained problem is simply taking the maximum $n$ for which the constraint is verified.

\begin{lemma}
 If $(u_{n+1}-u_n)_{n=0..n_\mathrm{max}}$ is increasing, there exists an increasing function $\gamma(x) \geqslant 0$ such that the solution maximising~\eqref{eq:ujplagamma_app} solves the constrained problem~\eqref{eq:constraints} and an increasing function $\gamma'(x) \geqslant  0$ such that the argument maximising problem~\eqref{eq:ujplagamma_othercriterion_app_supmat} solves the constrained problem~\eqref{eq:constraints_othercriterion_supmat}. 
\label{lem:utility_eq_app}
\end{lemma}
\begin{proof}

With the notation above, we have seen that $u_n$ is increasing, $v_n$ is decreasing, and $v_{n-1}-v_n$ is decreasing.
Furthermore, by hypothesis $u_n-u_{n-1}$ is increasing, which by definition of $v'_{n}$ means that $v'_{n-1}-v'_n$ is decreasing. In the following, we reason on $v_n$ but the argument is identical if $v_n$ is replaced by $v'_n$.


Let us fix $x>0$. The constrained problem is
$$\min_n v_n \quad\text{  subject to } u_n <= x$$
while the maximum utility problem can be rewritten as 
$$\min_n \left(v_n + \frac 1\gamma u_n\right).$$
Since $u_n$ is increasing, there is a highest $n_0=n(x)$ such that some configuration satisfies the constraint, i.e. such that $u_{n_0} <= x$. Since $v_n$ is decreasing, the solution of the constrained problem is found for $n=n_0$, for any configuration satisfying the constraint. We want to choose $\gamma$ such that the solution of the maximum utility problem is also at $n_0$, for a configuration satisfying the constraint. We show that we can choose $\gamma$ such  that taking $n\neq n_0$ leads to a larger value of $v_n+\frac 1\gamma u_n$.

From our hypotheses, we see that  the ratio of $u_{n+1} - u_{n}$ and $v_{n+1} - v_{n}$  is increasing.
For $n \leqslant n_0$, we will have 
$$v_n+\frac1\gamma u_n \geqslant v_{n_0}+\frac1\gamma  u_{n_0}$$
if we take 
\begin{equation} \label{eq:gammamore}
\gamma \geqslant \frac{u_{n_0}-u_{n_0-1}}{v_{n_0-1}-v_{n_0}}.
\end{equation}
Note that if $v_{n_0} - v_{n_0+1} = 0$, since $v_{n} - v_{n+1}$ is decreasing it means that for $n\geqslant n_0$ we have also $v_{n} - v_{n+1} = 0$, and we can always restrict the reasoning to $n_\mathrm{max}$ being the highest $n$ such that  $v_{n} - v_{n+1} \neq 0$.
For $n>n_0$, we will have 
$$v_{n}+\frac1\gamma u_n \geqslant v_{n_0} + \frac1\gamma u_{n_0}$$
if 
$$\gamma \leqslant \frac{u_{n_0+1}-u_{n_0}}{v_{n_0} - v_{n_0+1}}.$$

These two conditions can  be satisfied since
$$\frac{u_{n_0}-u_{n_0-1}}{v_{n_0-1}-v_{n_0}} \leqslant \frac{u_{n_0+1}-u_{n_0}}{v_{n_0} - v_{n_0+1}}.$$

 Choosing $\gamma$ between those two bounds gives it as an increasing function of $n_0$, thus as an increasing function of $x$.   
\end{proof}

As long as the sequence $(u_{n+1}^y-u_n^y)_{n=0...n_{\mathrm{max}}}$ is increasing, maximising the utility and the constrained problem have the same solutions. This is not guaranteed in the general case, but can be ensured under the following condition.
\begin{lemma}
If $\forall n>0$, $\exists i_0$, $\forall j = 1.. n-1$ , $\Theta_{i_0}^{n+1} \cap \Theta_j^{n-1} =\emptyset$, the sequence $(u_{n+1} - u_{n})_{n=1..n_{max}}$ is increasing. 
\label{lem:uincrease_app}
\end{lemma}
\begin{proof}
 Let us suppose that $\exists n\geqslant1$ such that $ u_{n+1} -  u_{n} <  u_{n} -  u_{n-1} $. Replacing by the explicit expression of $u_n$, the inequality is equivalent to 
\begin{align}
\sum\limits_{i=1}^{n}  \mathrm{TIP}_{\Theta_i^{n}}  - \sum\limits_{i=1}^{n+1}  \mathrm{TIP}_{\Theta_i^{n+1}}  < \sum\limits_{i=1}^{n-1}  \mathrm{TIP}_{\Theta_i^{n-1}}  - \sum\limits_{i=1}^{n}  \mathrm{TIP}_{\Theta_i^{n}} 
\label{eq:ineq_absurd}
\end{align}
By hypothesis, $\exists i_0$, $\forall j =1...n-1$, $\Theta_{i_0} \cap \Theta_j^{n-1} =\emptyset$, Eq.~\eqref{eq:ineq_absurd} can be written 
\begin{align}
\sum\limits_{i=1}^{n}  \mathrm{TIP}_{\Theta_i^{n}}  - \sum\limits_{i=1, i\neq i_0}^{n+1}  \mathrm{TIP}_{\Theta_i^{n+1}}  <  \mathrm{TIP}_{\Theta_{i_0}^{n+1}} + \sum\limits_{i=1}^{n-1}  \mathrm{TIP}_{\Theta_i^{n-1}}  - \sum\limits_{i=1}^{n}  \mathrm{TIP}_{\Theta_i^{n}} 
\end{align}
The term $ \mathrm{TIP}_{\Theta_{i_0}^{n+1}} + \sum\limits_{i=1}^{n-1}  \mathrm{TIP}_{\Theta_i^{n-1}} $ is a sum of $n$ $\mathrm{TIP}_{\Theta_i}$ with disjoint $\Theta_i$. By definition of $\Theta_i^n$, the right hand side of the inequality is less than or equal to 0 and the left hand side of the inequality is greater than or equal to 0, which is absurd.  

If  $\forall i = 1... n+1$, $\exists j = 1... n-1$, $\Theta_i^{n+1} \cap \Theta_j^{n-1} \neq \emptyset$. In that case, we also have $\forall i$, $\exists j \in \llbracket 1, n \rrbracket$ $\Theta_i^n \cap \Theta_j^{n-1} \neq \emptyset$, otherwise due to lemma~\ref{lem:find_app} this would lead to a contradiction. 
\end{proof}
If the condition of Lemma~\ref{lem:uincrease_app} is not satisfied one can find a counter example where $ u_{n+1}^y -  u_{n}^y <  u_{n}^y -  u_{n-1}^y $ and the equivalence of utility maximisation and optimisation with constraints is not guaranteed. 
 Finally, we have the proof of theorem~\ref{theorem},
\begin{proof}
Under the hypothesis of separability, by lemma~\ref{lem:uincrease_app},  $(u_{n+1}^y - u_{n}^y)_{n=1...n_{\mathrm{max}}}$ is increasing, and by lemma~\ref{lem:utility_eq_app}, we have the desired result. 
\end{proof}

\section{Other definition of missed detections}
\label{app:missed_detections_otherdef}
\label{app:g}

In this section, we show that the optimal procedure is similar if the false detections are defined as in~\cite{hara2022}, which is definition~\ref{def:misseddec1}.

For the missed detection definition used in the main text, Definition~\ref{def:misseddec2}, we recall Eq. (14) of the main text. Provided the separability condition (see Definition~\ref{def:separability}) is verified, a detection of a $n+1$-th planet is accepted if 
\begin{align}
    \mathrm{FIP}_{\Theta_{n+1}^{\star}} \leqslant \frac{\gamma}{\gamma+1}, 
    \label{eq:criterionFIP_fixed_app}
\end{align}
where $\Theta_{n+1}^{\star}$ is defined in Eq.~\eqref{eq:thetastar}.

For the missed detection Definition~\ref{def:misseddec1}, from Eq.~\eqref{eq:ujplagamma}, the $n+1$ component model has a greater utility than the $n$ component model if and only if 
\begin{align}
    \mathrm{FIP}_{\Theta_{n+1}^{\star}} \leqslant \gamma p(k\geqslant n+1\mid y) .
    \label{eq:criterionFIP}
\end{align} 
where $ p(k\geqslant n+1\mid y) $ is the probability that there are $n+1$ components or more. We add components until this criterion is violated.

Eq.~\eqref{eq:criterionFIP} means in particular that, for a given $\gamma$, the more components are claimed, the more stringent  the criterion to add a component becomes, since the term $p(k\geqslant n+1\mid y) $
gets smaller as $n$ increases. This contrasts with decisions based on a fixed threshold, for instance selecting a model with a Bayes factor greater than 150~\citep{kassraftery1995}.

  In that case, the criterion of Eq.~\eqref{eq:criterionFIP} is more stringent than the criterion of Eq. (14) of the main text (    $\mathrm{FIP}_{\Theta_{n+1}^{\star}} \leqslant \frac{\gamma}{\gamma+1}$). This behaviour is to be expected, because in the second case, the utility function has an extra penalization of missed detections. In Fig.~\ref{fig:mistakes_level_white_2}, we show the number of false detections and missed detections obtained with the new definition of missed detections on the simulation presented in Section 5.4 of the main text. In Fig.~\ref{fig:mistakes_roc_red_2}, we show the total number of mistakes as the detection threshold $\gamma$ varies. 

  \begin{figure}
\includegraphics[width=\linewidth]{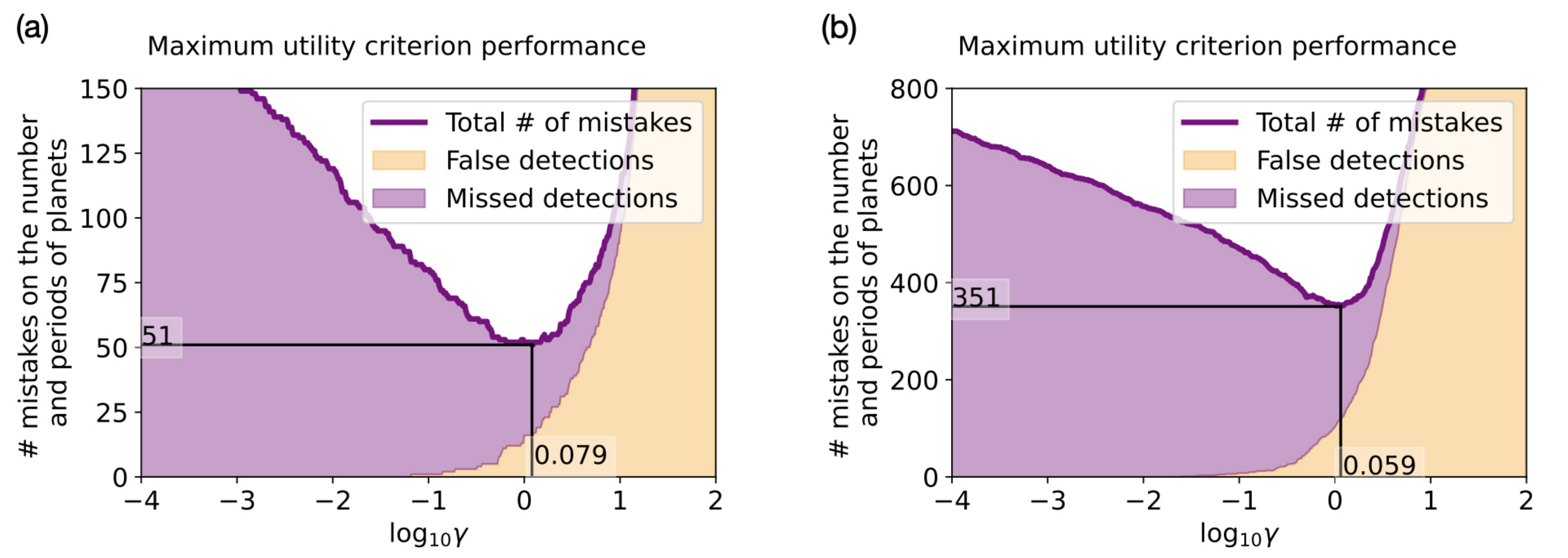}
\caption{Missed + false detections, in yellow and purple respectively, as a function of the detection threshold, $\log_{10} \gamma$ where $\gamma$ is defined in Eq.~\eqref{eq:criterionFIP}. (a) is obtained on the high SNR simulation and (b) the low SNR simulation. The black plain lines show where the minimum of mistakes is attained. }
\label{fig:mistakes_level_white_2}
\end{figure}
  \begin{figure}
\includegraphics[width=\linewidth]{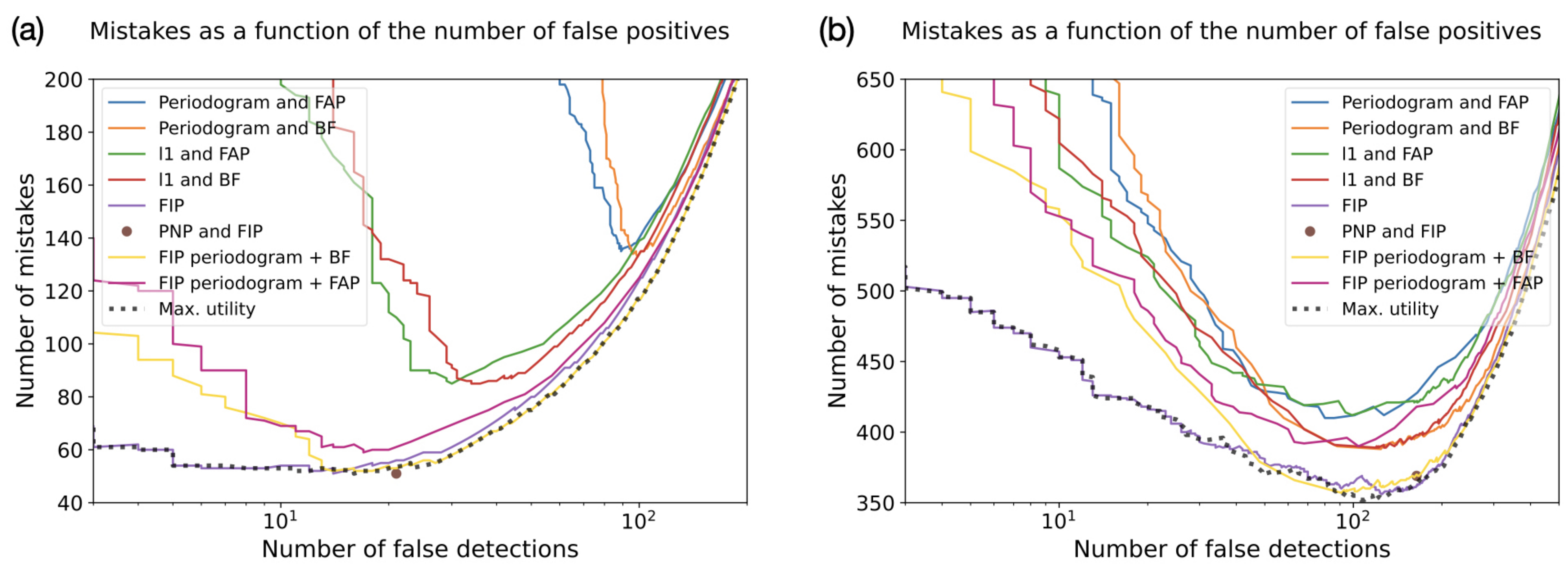}
\caption{Total number of mistakes (false+ missed detections) as a function of the number of false detections for the different analysis methods as their detection threshold vary, similarly to a ROC curve. (a) and (b) are obtained on the high and low SNR simulations, respectively.   The labels of the curves correspond to the methods used to find the candidate periods and to assess their significance. For instance, Periodogram and Bayes factor means that the candidate period is selected with a periodogram, and its significance assessed with a Bayes factor. }
\label{fig:mistakes_roc_red_2}
\end{figure}

\section{Data analysis methods}
\label{app:methods}

\label{app:h}

In Section 5.4 of the main text, we compare the performances of several methods, which we here describe in detail.  We first make explicit a few terms. The Bayes factor (BF) here refers to the ratio $p(y \mid k+1 )/p(y \mid k) $, where $p(y \mid k)$ is the Bayesian evidence (see Eq.~\eqref{eq:evidence}).
Bayes factor can also be seen in the framework of maximum utility~\citep{bernardosmith2009}. Assuming there are only two competing models, $M_0$ and $M_1$, for $i=0,1$ the utility function is 1 if $M_i$, is selected and, is the true model, and 0 if it is not. We define the posterior number of planets (PNP) as 
\begin{align}
  p(k \mid y) =  \frac{p(y \mid k) p(k)}{\sum\limits_{i=0}^{n_\mathrm{max}} p(y \mid i) p(i) } 
    \label{eq:pnp}
\end{align}
Furthermore, we here define the periodogram as in~\cite{delisle2019a}. This one is defined as a difference of log likelihoods of two models: a null hypothesis $H_0$ and a model $K_\omega$ with a sinusoidal component at frequency $\omega$. The periodogram at frequency $\omega$ is defined as
\begin{align}
\mathcal{P}(\omega) = \max_{\theta_{K_\omega}} \log p( y \mid \theta_{K_\omega}) - \max_{\theta_{H_0}} \log p( y \mid \theta_{H_0})  .
\end{align}
Denoting by $\mathcal{N}( x, \mathbf{V})$ a multivariate Gaussian  distribution of mean $ x$ and covariance $\mathbf{V}$
\begin{align}
H_0: \; y \sim & \; \mathcal{N}({0}, \mathbf{V_0}) \\
K_\omega: \; y = & A\cos \omega  t + B\sin \omega  t +  \epsilon, \; \; \epsilon \sim \mathcal{N}({0},  \mathbf{V_0})
\end{align}
where $V_0$ is the covariance matrix used to generate the noise in our simulated datasets. The periodogram is computed on a tightly spaced grid of frequencies between frequency frequency $\omega_{\mathrm{min}}$ and $\omega_{\mathrm{max}}$.
Denoting by $P^\star$ the maximum value of the periodogram, the false alarm probability (FAP) is defined as the probability that, knowing the null hypothesis, the maximum of the periodogram exceeds $P^\star$.  Denoting by $\Omega$ the grid of frequencies onto which the periodogram is computed,
\begin{align}
   \mathrm{FAP} =  p(\max_{\omega \in \Omega} \mathcal{P}(\omega) \geqslant  \mathcal{P}^\star \mid H_0)
\end{align}
We further use a $\ell_1$ periodogram as defined in~\cite{hara2017}. This algorithm searches for a representation of the data in the Fourier domain, penalizing the sum of their amplitudes with a $\ell_1$ norm, thus enhancing the sparsity of the representation.

In section 2.3 of the main text, we define the FIP periodogram, whose definition is reproduced here. 
We choose the sets $\Theta_i$s as frequency intervals of fixed length, chosen from a grid.  The interval $k$ of the grid is defined as $[\omega_k - \Delta \omega/2, \omega_k + \Delta \omega/2]$, where $ \Delta \omega = 2\pi/T_\mathrm{obs}$, $T_\mathrm{obs}$ is the total observation time span, and  $ \omega_k = k\Delta \omega / N_\mathrm{oversampling}$. We take $N_\mathrm{oversampling} = 5$.  
For each interval in the grid, we compute the marginal probability to have a planet whose frequency lies in the interval, and its FIP as defined in Eq.~\eqref{eq:fipdef}\footnote{It is in fact more computationally efficient to loop over the samples and determine in which frequency interval they lie, see the Python example \url{https://github.com/nathanchara/FIP}.}. 


\begin{itemize}
	 	\item Periodogram + FAP:  The periodogram is computed with the same grid of frequencies as the one used to generate the data (from 1.5 to 100 d) and the correct covariance matrix.  
	 	If the FAP is below a certain threshold fixed in advance, we add a cosine and sine function at the period of the maximum peak to the linear base model and recompute the periodogram. The planet is added if the FAP is below the FAP threshold. We do not search for a third planet.  
	 	\item Periodogram + Bayes factor: same as above, but here the criterion to add a planet is that the Bayes factor  is above a certain threshold.   The evidences (Eq.~\eqref{eq:evidence}) are computed with the  distributions used in the simulations, in particular the period is left free between 1.5 and 100 days.  
	 	\item  $\ell_1$-periodogram\footnote{The $\ell_1$ periodogram code is available at \url{https://github.com/nathanchara/l1periodogram}} + FAP : we compute the $\ell_1$ periodogram~\citep{hara2017} with the same grid of frequencies as the one used to generate the data (from 1.5 to 100 d). If the FAP of the maximum peak is below a certain threshold, it is added to the base model of unpenalized vectors, the $\ell_1$ periodogram is recomputed, the FAP of the maximum peak is assessed. If it is below a certain threshold, the a planet detection is claimed. We do not look for a third planet. 
	 	\item $\ell_1$-periodogram + Bayes factor: same as above, but here the criterion to add a planet is that the Bayes factor comparing n+1 versus n planet models is above a certain threshold. 
	 	\item FIP: We compute the FIP periodogram, and select the two highest peaks. We select a period if its corresponding FIP is below a certain threshold. 
	 	\item PNP + FIP periodogram: here, to select the number of planets we order the peaks of the FIP periodogram with increasing FIP. We select the number of peaks corresponding to the highest posterior number of planets as defined in Eq.~\eqref{eq:pnp}. 
	 	\item FIP  periodogram + Bayes factor : the periods are selected as the highest peaks of the FIP periodogram and the number of planets is selected with the Bayes factor. This procedure is very close to~\cite{gregory2007} except that we use the FIP periodogram instead of the marginal distribution of periods for each planets. We do not take the approach of~\cite{gregory2007} to select the periods as nested sampling algorithms such as \textsc{polychord} tend to swap the periods of planets, such that marginal distributions are typically multi-modal. 
	 	\item  FIP periodogram + FAP : the periods are selected as the highest peaks of the FIP periodogram and the number of planets is selected with the false alarm probability. 
\end{itemize}

\bibliographystyle{imsart-nameyear}
\bibliography{biblio}

\end{document}